\documentclass[letterpaper, 11pt]{article}
\usepackage{makeidx}  
\usepackage{algorithm}
\usepackage{algorithmicx,algpseudocode}
\usepackage{amsmath,amsfonts,amsthm,amssymb}
\usepackage[left=1in,top=1in,right=1in,bottom=1in]{geometry}
\usepackage{xcolor}
\usepackage{enumitem}

\usepackage{ifpdf}
\ifpdf    
\usepackage{hyperref}
\else    
\usepackage[hypertex]{hyperref}
\fi

\newtheorem{theorem}{Theorem}[section]
\newtheorem{lemma}[theorem]{Lemma}

\newtheorem{corollary}[theorem]{Corollary}
\newtheorem{conjecture}[theorem]{Conjecture}

\newtheorem{claim}[theorem]{Claim}

\theoremstyle{remark}

\newtheorem{assumption}[theorem]{Assumption}

\theoremstyle{definition}
\newtheorem{definition}[theorem]{Definition}

\newcommand{\Xomit}[1]{}

\newcommand{\opt}{{\sf OPT}}

\newcommand{\eps}{\varepsilon}

\newcommand{\prob}{\mathrm{Prob}}

\newcommand {\brc}   [1] {\left(#1\right)}
\newcommand {\Probab}  [1] {\Pr \brc{#1 }}
\newcommand {\set}   [1] {\left\{ #1 \right\}}
\DeclareMathOperator {\cost}  {cost}

\newcommand{\sopt}{S_{\opt}}
\newcommand{\topt}{T_{\opt}}

\begin{document}

\title{Minimizing the Union: Tight Approximations for Small Set Bipartite Vertex Expansion}

\author{Eden Chlamt\'a\v{c}\thanks{Partially supported by ISF grant 1002/14.}\\Ben Gurion University \and Michael Dinitz\thanks{Supported by NSF awards 1464239 and 1535887.}\\Johns Hopkins University \and Yury Makarychev\thanks{Supported by NSF awards CAREER CCF-1150062 and IIS-1302662.}\\TTIC}

\begin{titlepage}
\def\thepage{}
\maketitle
\begin{abstract} 
In the Minimum $k$-Union problem (M$k$U) we are given a set system with $n$ sets and
are asked to select $k$ sets in order to minimize the size of their
union.  Despite being a very natural problem, it has received
surprisingly little attention: the only known approximation algorithm
is an $O(\sqrt{n})$-approximation due to [Chlamt\'a\v{c} et al
APPROX~'16].  This problem can also be viewed as the bipartite version
of the Small Set Vertex Expansion problem (SSVE), which we call the
Small Set Bipartite Vertex Expansion problem (SSBVE).  SSVE, in which
we are asked to find a set of $k$ nodes to minimize their vertex
expansion, has not been as well studied as its edge-based counterpart
Small Set Expansion (SSE), but has recently received significant
attention, e.g.~[Louis-Makarychev APPROX '15].  However, due to the
connection to Unique Games and hardness of approximation the focus has
mostly been on sets of size $k = \Omega(n)$, while we focus on the
case of general $k$, for which no polylogarithmic approximation is
known.

We improve the upper bound for this problem by giving an
$n^{1/4+\eps}$ approximation for SSBVE for any constant $\eps > 0$.
Our algorithm follows in the footsteps of Densest $k$-Subgraph (DkS)
and related problems, by designing a tight algorithm for random
models, and then extending it to give the same guarantee for arbitrary
instances.  Moreover, we show that this is tight under plausible
complexity conjectures: it cannot be approximated better than
$O(n^{1/4})$ assuming an extension of the so-called ``Dense versus
Random" conjecture for DkS to hypergraphs.

In addition to conjectured hardness via our reduction, we show that
the same lower bound is also matched by an integrality gap for a
super-constant number of rounds of the Sherali-Adams LP hierarchy,
and an even worse integrality gap for the natural SDP relaxation.
Finally, we note that there exists a simple bicriteria $\tilde O(\sqrt{n})$
approximation for the more general SSVE problem (where no non-trivial approximations
were known for general $k$).
\end{abstract}

\end{titlepage}


\section{Introduction}
Suppose we are given a ground set $U$, a set system $\mathcal S \subseteq 2^U$ on $U$, and an integer $k$.  One very natural problem is to choose $k$ sets from $\mathcal S$ in order to maximize the number of elements of $U$ that are covered.  This is precisely the classical Maximum Coverage problem, which has been well-studied and is known to admit a $(1-1/e)$-approximation (which is also known to be tight)~\cite{Feige98}.  But just as natural a problem is to instead choose $k$ sets in order to \emph{minimize} the number of elements of $U$ that are covered.  This is known as the \emph{Minimum $k$-Union} problem (M$k$U), and unlike Maximum Coverage, it has not been studied until recently (to the best of our knowledge), when an $O(\sqrt{m})$-approximation algorithm was given by Chlamt\'a\v{c} et al.~\cite{CDKKR16} (where $m$ is the number of sets in the system).  This may in part be because M$k$U seems to be significantly harder than Maximum Coverage: when all sets have size exactly $2$ then M$k$U is precisely the Smallest $m$-Edge Subgraph problem (S$m$ES), which is the minimization version of the well-known Densest $k$-Subgraph (D$k$S) problem and is thought to be hard to approximate better than a polynomial.  M$k$U is the natural hypergraph extension of S$m$ES.

M$k$U is also related to another set of problems which are thought to be hard to approximate: problems similar to \emph{Small Set Expansion}.  Given an instance of M$k$U, we can construct the obvious bipartite graph in which the left side represents sets, the right side represents elements, and there is an edge between a set node and an element node if the set contains the element.  Then M$k$U is clearly equivalent to the problem of choosing $k$ left nodes in order to minimize the size of their neighborhood.  We call this the \emph{Small Set Bipartite Vertex Expansion} (SSBVE) problem, and is the way we will generally think of M$k$U throughout this paper.  This is the bipartite version of the Small Set Vertex Expansion (SSVE) problem (in which we are given an arbitrary graph and are asked to choose $k$ nodes to minimize the size of their neighborhood), which is in turn the vertex version of the Small Set Expansion (SSE) problem (in which we are asked to choose a set of $k$ nodes to minimize the number of \emph{edges} with exactly one endpoint in the set).  SSE, and SSVE to a lesser extent, have been extensively studied due to their connection to other hard problems (including the Unique Games Conjecture), but based on these connections have generally been considered only when $k$ is relatively large (in particular, when $k = \tilde \Omega(n)$).

\subsection{Random models and worst case approximations}

Given the immediate connection to D$k$S and S$m$ES, it is natural not only to examine the techniques used for those problems, but the general algorithmic framework -- the so called ``log-density" framework --  developed for these problems. This approach, first introduced in~\cite{BCCFV10}, can be summarized as follows.  Begin by considering the problem of distinguishing between a random structure (for D$k$S/S$m$ES, a random graph) and a random structure in which a small, statistically similar solution has been planted. Several algorithmic techniques (both combinatorial and LP/SDP based) fail to solve such distinguishing problems~\cite{BCCFV10,FS97-sdpgap,BCVGZ12}, and thus the gap (in the optimum) between the random and planted case can be seen as a natural lower bound for approximations. To match these lower bounds algorithmically, one develops robust algorithms for this distinguishing problem for the case when the planted solution \emph{does have} statistically significant properties that would not appear in a pure random instance, and then, with additional work, adapts these algorithms to work for worst-case instances while guaranteeing the same approximation.

While this framework was new when introduced in~\cite{BCCFV10}, the actual technical tools to adapt algorithms for random planted instances of D$k$S and S$m$ES\footnote{In~\cite{CDK12}, most of the technical work on S$m$ES focused on certain strong LP rounding properties, however the basic combinatorial algorithm was similar to the algorithm for D$k$S in~\cite{BCCFV10}.} to the adversarial setting were not particularly complex. While a great many problems have strong hardness based on reductions from D$k$S/S$m$ES (most importantly Label Cover and the great many problems with hardness reductions from Label Cover), tight approximations in this framework have not been achieved for any problem since~\cite{CDK12}. This despite the fact that a number of algorithms for such problems can be seen as a partial application of the same techniques (e.g.~\cite{CHK11,MM13,MMS14}). The reason is that the approach only provides a general framework. As with any tool (e.g. SDPs), to apply it, one must deal with the unique technical challenges offered by the specific problem one is attacking.

As we shall see, we are able to successfully apply this approach to M$k$U/SSBVE to achieve tight approximations, making this only the third complete application\footnote{It has also been applied to Label Cover in a recent submission which includes the first author, though only in the semirandom setting.} of the
framework, and the first one to overcome technical obstacles which deviate significantly from~\cite{BCCFV10}.

\subsection{Problem Definitions and Equivalence}

We study M$k$U/SSBVE, giving an improved upper bound which is tight in the log-density framework. We also strengthen the lower bounds from conjectures for random models in this framework by showing that they are matched by integrality gaps for the  Sherali-Adams LP hierarchy, and that the natural SDP relaxation has an even worse integrality gap.

Slightly more formally, we will mostly study the following two problems.  Given a graph $G = (V, E)$ and a subset $S \subseteq V$, define the neighborhood $N(S) = \{v : \exists \{u,v\} \in E \land u \in S\}$.  In a bipartite graph $G=(U,V,E)$, the \emph{expansion} of a set $S \subseteq U$ (or $S \subseteq V$) is $|N(S)| / |S|$.

\begin{definition} \label{def:MkU}
In the \emph{Minimum $k$-Union} problem (M$k$U), we are given a universe $V$ of $n$ elements and a collection of $m$ sets $\mathcal S \subseteq 2^U$, as well as an integer $k \leq m$.  The goal is to return a collection $T \subseteq \mathcal S$ with $|T| = k$ in order to minimize $\cup_{S \in T} S$.
\end{definition}

\begin{definition} \label{def:SSBVE}
In the \emph{Small Set Bipartite Vertex Expansion} problem (SSBVE) we are given a bipartite graph $(U, V, E)$ with $n = |U|$ and $n' = |V|$ and an integer $k \leq |U|$.  The goal is to return a subset $S \subseteq U$ with $|S| = k$ minimizing the expansion $|N(S)| / k$ (equivalently, minimizing $|N(S)|$).
\end{definition}

The next lemma is obvious, and will allow us to use the two problems interchangeably.
\begin{lemma} \label{lem:MkU-SSBVE}
M$k$U and SSBVE are equivalent: There is an $f(m)$-approximation for M$k$U if and only if there is an  $f(n)$-approximation for SSBVE.\footnote{Note that it is not immediately clear why $n(=|U|)$ is the natural input parameter for an approximation guarantee for SSBVE. This is discussed in Section~\ref{sec:discussion}.}
\end{lemma}

While the two problems are equivalent, in different contexts one may be more natural than the other. As we have noted, and will also discuss later, all previous work on these problems has been through the lens of M$k$U, which is especially natural in the graph case when studying D$k$S and S$m$ES. Moreover, the random distinguishing models used in our conjectured lower bounds are based on random graphs and their extension to $O(1)$-uniform hypergraphs, and thus are best understood as applied to M$k$U. However, our approximation algorithm is much more easily explained as an algorithm for SSBVE. Thus, we will use the M$k$U notation when discussing random models and conjectured hardness, and SSBVE when describing our algorithm.

\subsection{Our Results and Techniques}

As mentioned, we look at random distinguishing problems of the form studied by Bhaskara et al.~\cite{BCCFV10} and Chlamt\'a\v{c} et al.~\cite{CDK12}. Define the \emph{log-density} of a graph on $n$ nodes to be $\log_n(D_{avg})$, where $D_{avg}$ is the average degree.  One of the problems considered in~\cite{BCCFV10,CDK12} is the \textsc{Dense vs Random} problem, which is parameterized by $k$ and constants $0 < \alpha, \beta < 1$: Given a graph $G$, distinguish between the following two cases: 1) $G = G(n,p)$ where $p = n^{\alpha-1}$ (and thus the graph has log-density concentrated around $\alpha$), and 2) $G$ is adversarially chosen so that the densest $k$-subgraph has log-density $\beta$ where $k^{\beta} \gg pk$ (and thus the average degree inside this subgraph is approximately $k^{\beta}$).  The following conjecture was explicitly given in~\cite{CDK12}, and implies that the known algorithms for D$k$S and S$m$ES are tight:

\begin{conjecture}\label{con:DvR}
For all $0 < \alpha < 1$, for all sufficiently small $\eps > 0$, and for all $k \leq \sqrt{n}$, we cannot solve \textsc{Dense vs Random} with log-density $\alpha$ and planted log-density $\beta$ in polynomial time (w.h.p.) when $\beta \leq \alpha - \eps$.
\end{conjecture}

This conjecture can quite naturally be extended to hypergraphs.  Let $\mathcal G_{n,p,r}$ denote the distribution over $r$-uniform hypergraphs obtained by choosing every subset of cardinality $r$ to be a hyperedge independently with probability $p$.  Define the \textsc{Hypergraph Dense vs Random} problem as follows, again parameterized by $k$ and constants $0 < \alpha,\beta < r-1$.  Given an $r$-uniform hypergraph $G$ on $n$ nodes, distinguish between the following two cases: 1) $G = \mathcal G_{n,p,r}$ where $p = n^{\alpha-(r-1)}$ (and thus the log-density is concentrated around $\alpha$), and 2) $G$ is adversarially chosen so that the densest subhypergraph on $k$ vertices has log-density $\beta$ (and thus the average degree in the subhypergraph is $k^{\beta}$). 

\begin{conjecture} \label{con:HDvR}
For all constant $r$ and $0 < \beta < r-1$, for all sufficiently small $\eps > 0$, and for all $k$ such that $k^{1+\beta} \leq n^{(1+\alpha)/2}$, we cannot solve \textsc{Hypergraph Dense vs Random} $\alpha$ and planted log-density $\beta$ in polynomial time (w.h.p.) when $\beta < \alpha - \eps$.
\end{conjecture}

An easy corollary of this conjecture (proved in the appendices by setting parameters appropriately) is that for any constant $\eps > 0$, there is no polynomial-time algorithm which can distinguish between the two cases from \textsc{Hypergraph Dense vs Random} when the gap between the M$k$U objective function in the two instances is $\Omega(m^{1/4 - \eps})$.  By transforming to SSBVE, we get a similar gap of $\Omega(n^{1/4 - \eps})$.  This also clearly implies the same gap for the worst-case setting.

Complementing this lower bound, we indeed show that in the random planted setting, we can appropriately modify the basic structure of the algorithm of~\cite{BCCFV10} and achieve an $O(n^{1/4+\eps})$-approximation for any constant $\eps>0$, matching the above lower bound for this model. However, our main technical contribution is an algorithm which matches this guarantee in the \emph{worst case} setting 
(thus improving over \cite{CDKKR16}, who gave a $O(\sqrt{m})$-approximation for M$k$U, i.e.\ an $O(\sqrt{n})$-approximation for SSBVE).

\begin{theorem} \label{thm:main}
For any constant $\epsilon > 0$, there is a polynomial-time $O(n^{1/4 + \epsilon})$-approximation algorithm for SSBVE.
\end{theorem}

We prove Theorem~\ref{thm:main} in Section~\ref{sec:algo}. This implies an $O(m^{1/4 + \epsilon})$-approximation for M$k$U (recall that $m$ is the number of sets in an M$k$U instance).  It is natural to wonder whether we can instead get an approximation depending on $n$ (the number of elements).  Unfortunately, we show in Appendix~\ref{sec:random-app} that this is not possible assuming Conjecture~\ref{con:HDvR}.

While our aim is to apply the framework of~\cite{BCCFV10} to SSBVE, we note that SSBVE and D$k$S (or the minimization version S$m$ES) differ from each other in important ways, making it impossible to 
straightforwardly apply the ideas from~\cite{BCCFV10} or~\cite{CDK12}.  
The asymmetry between $U$ and $V$ in SSBVE
requires us to fundamentally change our approach; 
loosely speaking, in SSBVE, we are looking not for an arbitrary dense subgraph of a bipartite graph, but rather
for a dense subgraph of the form $S \cup N(S)$ (where $S$ is a subset of $U$ of size $k$),
since once we choose the set $S\subset U$ we \emph{must} take into account all neighbors of $S$ in $V$.

For example, suppose that there are $k$ nodes in $U$ with degree $r$ whose neighborhoods overlap on some set of $r-1$ nodes of $V$, but each of the $k$ nodes also has one neighbor that is not shared by any of the others.   Then a D$k$S algorithm (or any straightforward modification) might return the bipartite subgraph induced by those $k$ nodes and their $r-1$ common neighbors.  But even though the intersection of the $k$ neighborhoods is large, making the returned subgraph very dense, their \emph{union} is \emph{much} larger (since $k$ could be significantly larger than $r$).  So taking those $k$ left nodes as our SSBVE solution would be terrible, as would any straightforward pruning of this set.

This example shows that we cannot simply use a D$k$S algorithm, and there is also no reduction which lets us transform an arbitrary SSBVE instance into a D$k$S instance where we could use such an algorithm.  Instead, we must fundamentally change the approach of~\cite{BCCFV10} to take into account the asymmetry of SSBVE.  One novel aspect of our approach is a new asymmetric pruning idea which allows us to isolate a relatively small set of nodes in $V$ which will be responsible for collecting all of the ``bad" neighbors of small sets which would otherwise have small expansion. Even with this tool in place, we still need to trade off a number of procedures in each step to ensure that if the algorithm halts it will return a set that is both small and has small expansion (ignoring the pruned set on the right).

In addition to the conditional lower bound guaranteed by the log-density framework (matching our upper bound), we can show unconditional lower bounds 
 against certain types of algorithms: those that depend on Sherali--Adams (SA) lifts of the basic LP, or those that depend on the basic SDP relaxation.  We do this by showing integrality gaps which match or exceed our upper bound.

\begin{theorem} \label{thm:SDP-main}
The basic SDP relaxation of SSBVE has an integrality gap of $\tilde \Omega(n^{1/2})$.
\end{theorem}

\begin{theorem} \label{thm:LP-main}
When $r = O(\eps \log n / \log \log n)$, the integrality gap of the $r$-round Sherali--Adams relaxation of SSBVE is $n^{1/4 - O(\eps)}$.
\end{theorem}

We show these integrality gaps in Section~\ref{sec:gaps}. In our SA gap construction, we use the general framework from~\cite{BCCFV10},
where they present a SA integrality gap for Densest $k$-Subgraph. However, we have to make significant changes
to their gap construction, since there is an important difference between D$k$S and SSBVE
  --- the former problem does not have hard constraints, while
the latter one does. Specifically, in SSBVE, if a vertex $u\in U$ belongs to the solution $S$, then
every neighbor $v$ of $u$ is in the neighborhood of $S$. This means, in particular, that in the SA solution, variable $x_{\{u\}}$
(the indicator variable for the event that $u\in S$) must be exactly equal to $x_{\{u\} \cup B}$
(the indicator variable for the event that $u\in S$ and each vertex in $B$ is in the neighborhood of $S$)
for every subset $B$ of neighbors of $u$. 
More generally, if $A \subset U \cup V$ and $B\subset N(A\cap U)$, then $x_{A}$ must be equal to $x_{A\cup B}$.
However, in the SA integrality gap construction for D$k$S, $x_{A \cup B}$ is \textit{exponentially} smaller than $x_{A}$:
$x_{A \cup B} \ll e^{-\Omega(|B\setminus A| \log\log n)} x_{A}$; this inequality is crucial because it guarantees
 that the SA solution for D$k$S is feasible.
In our construction, we have to very carefully define variables $x_A$ in order to ensure that,
on one hand, $x_{A} = x_{A\cup B}$ and, on the other hand, the solution is feasible.

While not the main focus of this paper, we also give an improved approximation for the \emph{Small Set Vertex Expansion} problem (SSVE) and explore its relationship to SSBVE.  In SSVE we are given a graph $G = (V, E)$ and an integer $k$, and the goal is to find a set $S \subseteq V$ with $|S| \leq k$ such that $|N(S) \setminus S|/|S|$ is minimized.  Louis and Makarychev~\cite{LM14} gave a polylogarithmic bicriteria approximation for this problem for $k \geq n/\mathrm{polylog}(n)$, but to the best of our knowledge there are no current bounds for general $k$.  We give the first nontrivial upper bound, which is detailed in Section~\ref{sec:SSVE}:

\begin{theorem} \label{thm:SSVE-main}
There is an $(1+\eps, \tilde O(\sqrt{n}))$-approximation for SSVE (the algorithm chooses a set of size at most $(1+\eps) k$ but is compared to the optimum of size at most $k$).
\end{theorem}

Finally, we note that as written, SSVE is an ``at most" problem, where we are allowed any set of size \emph{at most} $k$ but the set size appears in the denominator.  On the other hand, our definition of SSBVE requires picking \emph{exactly} $k$ nodes.  We could define an equivalent exact problem for SSVE, where we require $|S| = k$, and an equivalent at most problem for SSBVE, where we allow sets of size at most $k$ but instead of minimizing $|N(S)|$  minimize $|N(S)| / |S|$.  It is straightforward to show that up to a logarithmic factor the at most and the exact versions of the problems are the same, and the following lemma appears in~\cite{CDKKR16} for SSBVE (the equivalent for SSVE is just as easy).
\begin{lemma} \label{lem:at-most-exact}
An $f$-approximation algorithm for the at-most version of SSBVE (SSVE) implies an $\tilde O(f)$-approximation for the exact version of SSBVE (SSVE), and vice versa.
\end{lemma}
Hence we will generally feel free to consider one version or the other depending on which is easier to analyze.

\subsection{Related Work}
As mentioned, SSBVE is a generalization of the Smallest $m$-Edge Subgraph problem~\cite{CDK12}, which is the minimization version of Densest $k$-Subgraph~\cite{FKP01,BCCFV10}.  The best-known approximation for S$m$ES is $O(n^{3-2\sqrt{2} + \eps})$~\cite{CDK12}, while the best-known approximation for D$k$S is $O(n^{1/4 + \eps})$~\cite{BCCFV10}.

The immediate predecessor to this work is~\cite{CDKKR16}, which provided an $O(\sqrt{n})$-approximation for SSBVE ($O(\sqrt{m})$ for M$k$U) using relatively straightforward techniques.  SSVE was also studied by Louis and Makarychev~\cite{LM14}, who provided a polylogarithmic approximation when $k$ is very close to $n$ (namely, $k \geq n/\mathrm{polylog}(n)$).  To the best of our knowledge, no approximation was known for SSVE for general $k$.

While defined slightly differently, the maximization version of M$k$U, the Densest $k$-Subhypergraph problem (D$k$SH), was defined earlier by Applebaum~\cite{Applebaum13} in the context of cryptography: he showed that if certain one way functions exist (or that certain pseudorandom generators exist) then D$k$SH is hard to approximate within $n^{\epsilon}$ for some constant $\epsilon > 0$.  Based on this result, D$k$SH and M$p$U were used to prove hardness for other problems, such as the $k$-route cut problem~\cite{CMVZ16}.  He also explicitly considered something similar to \textsc{Hypergraph Dense vs Random}, but instead of distinguishing between a random instance and an adversarial instance with essentially the same log-density, he considered the problem of distinguishing a random instance from a random instance which has a planted dense solution, and then where every hyperedge not in the planted solution is randomly removed with some extra probability.  He showed that even this problem is hard if certain pseudorandom generators exist.


\section{Approximation algorithm for random planted instances}\label{sec:algo}

Before describing our algorithm for SSBVE (in either the random or worst case setting), let us mention a key tool, which was also used in~\cite{CDKKR16}, and can be seen as a straightforward generalization of an LP (and rounding) in~\cite{Charikar00}:
\begin{lemma}\label{lem:charikar} There is a polynomial time algorithm which exactly solves the Least Expanding Set problem, in which we are given a bipartite graph $(U,V,E)$, and wish to find a set $S\subseteq U$ with minimum expansion $|N(S)|/|S|$.
\end{lemma}
Note that in the Least Expanding Set problem there is no constraint on the cardinality of the set $S$, and so the above lemma has no immediate implication for SSBVE.

Recall that Lemma~\ref{lem:at-most-exact} shows the equivalence of the ``at most" and ``exact" versions of SSBVE interchangeably. Thus, throughout this section and the next, we will use these two version interchangeably.

To understand our approximation algorithm, let us first consider the following random model: For constants $\alpha,\beta,\gamma\in(0,1)$ where $\gamma<\beta$, let $G=(U,V,E)$ be a random bipartite graph where $|U|=n$, $|V|=n^{\beta}$, and let $k=n^{1-\alpha}$. The edges $E$ are defined by first choosing $r=\log n$ neighbors in $V$ independently at random for every $u\in U$, and then choosing subsets $S\subset U$ and $T\subset V$ of size $|S|=k=n^{1-\alpha}$ and $|T|=n^{\gamma}$ independently at random, and for every node $u\in S$, remove its $r$ neighbors and now sample them uniformly at random from $T$.

Suppose first that $\gamma\geq(1-\alpha)(\beta-\eps)$. That is, $|T|\leq n^{(1-\alpha)(\beta-\eps)}=k^{\beta-\eps}$. This means that the log-density gap between the random graph and the planted subgraph is small.  Note that any arbitrary $k$-subset of $U$ will expand to at most $\min\{k\log n,|V|\}=\tilde O(n^{\min\{1-\alpha,\beta\}})$ vertices in $V$, compared to the optimum (planted) solution $S$ which expands to $n^\gamma$ vertices. In this case, up to a logarithmic factor, choosing any such set gives us  an approximation ratio of 
\begin{align*}n^{\min\{1-\alpha,\beta\}-\gamma}&\leq n^{\min\{1-\alpha,\beta\}-(1-\alpha)(\beta-\eps)}\\&=n^{\min\{(1-\alpha)(1-\beta),\alpha\beta\}+\eps(1-\alpha)}.\end{align*}
It is easily seen that this expression is maximized for $\beta=1-\alpha$, giving an approximation ratio of $n^{(1-\alpha)(\alpha+\eps)}=n^{\alpha(1-\alpha)}k^{\eps}$, which is clearly at most $n^{1/4+\eps}$. Since this is the approximation ratio we are aiming for, we may focus on the case when $\gamma\leq(1-\alpha)(\beta-\eps)$. For simplicity, let us look at the tight case, when $\beta=1-\alpha$, and $\gamma=(\beta-\eps)(1-\alpha)$. That is, $|S|=|V|=k=n^{1-\alpha}$, and $|T|=k^{1-\alpha-\eps}$. The expansion of the hidden subset is $|T|/|S|=k^{-\alpha-\eps}$. 

Consider as a simple example the case of $\alpha=\frac12$. For simplicity, let us think of the left-degree $r$ as some large constant rather than $\log n$. In this case $k=|V|=\sqrt{n}$, and every vertex in $V$ has $\Theta(\sqrt{n})$ neighbors in $U$. Choosing a vertex $v\in T$ (say, by guessing all possible vertices in $V$) gives us the following: The neighborhood $N(v)$ has size $\Theta(\sqrt{n})$, however because of the planted solution, $\Omega(k^{1/2+\eps})$ of the vertices in $N(v)$ also belong to $S$. We know that $N(v)\cap S$ expands to $T$ which has size $k^{1/2-\eps}$. That is, it has expansion $O(k^{-2\eps})$. 
 Thus by Lemma~\ref{lem:charikar} applied to the subgraph induced on $(N(v),V)$, we can find a set $S\subseteq N(v)$ with at most this expansion, which gives an approximation ratio of $O(k^{-2\eps}/k^{-1/2-\eps})=O(k^{1/2-\eps})\ll n^{1/4}$, and moreover $|S|\leq|N(v)|=O(k)$, so we are done.

Now let us consider the general case. Suppose $\alpha=p/q$ for some relatively prime $q>p>0$. Note that the degree of every vertex in $V$ is tightly concentrated around $r|U|/|V|=\Theta(n^{p/q})$ and the $S$-degree of every vertex in $T$ (the cardinality of its neighborhood intersected with $S$) is concentrated around $r|S|/|T|=\Theta(k^{p/q+\eps})$.

Following the approach of~\cite{BCCFV10} for D$k$S, we can think of an algorithm which inductively constructs all possible copies of a caterpillar with fixed leaves which are chosen (guessed) along the way. In our case, the caterpillar is similar, but not identical to the corresponding caterpillar used in~\cite{BCCFV10}. Every step in the construction of the caterpillar corresponds to an interval of the form $((j-1)\alpha,j\alpha)$. The overall construction is as follows:

\begin{itemize}[itemsep=-1pt, topsep=4pt, partopsep=0pt]
  \item First step (step 1): This step corresponds to the interval $(0,\alpha)$. In this step we guess a vertex $v$, and add an edge, which forms the first edge of the ``backbone" (the main path in the caterpillar).
  \item Final step (step q): This step corresponds to the interval $((q-1)\alpha,q\alpha)=(p-\alpha,p)$. In this step we add a final edge to the backbone.
  \item Intermediate step: For every $j=2,\ldots,q-1$, step $j$ corresponds to the interval $((j-1)\alpha,j\alpha)$. If this interval contains an integer, choose a new vertex and attach an edge (a ``hair") from it to the (currently) last vertex in the backbone. Otherwise, extend the backbone by \emph{two} edges.\footnote{This is the only difference from the caterpillar of~\cite{BCCFV10}, where they add one edge to the backbone in this case.}
\end{itemize}

Note that if we start the caterpillar with a vertex in $V$, then the next backbone vertex will be in $U$, and since we then add two backbone edges each time, the current backbone vertex will always be in $U$ (until the last step), and all leaves guessed along the way will be in $V$. 

How do we turn this into an algorithm for random planted instances as above? Start by guessing a vertex $v\in T$ (by exhaustive enumeration over all vertices in $V$), and start with set $W=\{v\}$. Whenever the caterpillar construction adds an edge to the backbone, update the set $W$ to $N(W)$. Whenever the construction caterpillar adds a hair, guess a new vertex $v'\in T$ (again we can assume we guess a vertex in $T$ by exhaustive enumeration) and update the set $W$ to $W\cap N(v')$. Do this for all edges except for the last edge in the caterpillar. Note that whenever we have a ``backbone step", the caterpillar construction adds two edges to the backbone, so if we had a set $W\subseteq U$ (as we do at the beginning of every intermediate step), 
 we end up with the set $N(N(W))\subseteq U$.

An easy inductive argument using concentration in random graphs shows that after every step $j\in[q-1]$, w.h.p.\ $W$ is a subset of $U$ of size $\Theta(n^{j\alpha-\lfloor j\alpha\rfloor})$, and as long as the vertices we guessed were indeed in $T$, for sufficiently small $\eps$ (at most $1/q^2$), the set $W\cap S$ has size $k^{j\alpha-\lfloor j\alpha\rfloor +j\eps}$.

In particular, right after step $q-1$, we have $|W|=\Theta(n^{1-\alpha})=\Theta(k)$ and $|W\cap S|=\Theta(k^{\alpha+(q-1)\eps})$. At this time the set $W\cap S$ expands to $T$ which has size $k^{\alpha-\eps}$, and so $W\cap T$ has expansion $O(k^{-q\eps})$. Thus, by Lemma~\ref{lem:charikar} applied to $W$, we can find a set $S\subseteq W$ (of size at most $|W|=O(k)$) with at most this expansion, giving an approximation ratio of $O(k^{-q\eps}/k^{\alpha-1-\eps})=O(k^{1-\alpha-(q-1)\eps})\ll k^{1-\alpha}.$ While we cannot guarantee that $k$ will be exactly $k=n^{p/q}$ for some reasonably small constants $q>p>0$, we can guarantee that we will not lose more than a $k^{O(1/q)}$-factor by running this algorithm on nearby values of $k$, and so at least in this random model, when $k=n^{\alpha}$, we can always achieve an approximation guarantee of $k^{(1-\alpha)+\eps}=n^{\alpha(1-\alpha)+\eps}\leq n^{1/4+\eps}$. Our main technical contribution is translating this overall framework into an algorithm that achieves the same approximation guarantee for worst-case instances, which we do in the next section.

\section{Approximation algorithm for worst case instances}
In this section, we show the desired $n^{1/4+\eps}$ approximation for worst case instances. As planned, we follow the above caterpillar structure for random planted instances, so that (after some preprocessing) at every step either set sizes behave like in the random planted setting, or we can find set with small expansion. We start with some preprocessing which will be useful in the analysis.

\subsection{Preprocessing}

Using standard bucketing and subsampling techniques, and the trivial algorithm (taking an arbitrary $k$-subset when there is no log-density gap, or only a small gap), we can restrict our attention to a fairly uniform setting:

\begin{lemma}\label{lem:preprocessing} Suppose for every sufficently small $\eps>0$ and for all integers $0<p<q=O(1/\eps)$ there is some constant $c=c(\eps,p,q)>0$ such that we can obtain an $O(k^{p/q+(1-c)\eps})$-approximation for SSBVE on instances of the following form:
\begin{itemize}[itemsep=-1pt, topsep=4pt, partopsep=0pt]
  \item Every vertex in $U$ has the same degree $r$.
  \item The size $|N(\sopt)|$ of the neighbor set of the least expanding $k$-subset of $U$ is known, and thus so is the average degree from this set back to the least expanding $k$-set, $d=kr/|N(\sopt)|$.
  \item We have $k=n^{1-p/q}$, 
  \item The optimum average back-degree satisfies $d=k^{p/q+\eps}$. 
\end{itemize}
Then there is an $O(n^{1/4+\eps})$-approximation for SSBVE for every sufficiently small $\eps>0$.
\end{lemma}

We defer the proof of this lemma to Appendix~\ref{sec:preprocessing-app}. From now on, we will assume the above setting, and denote $\alpha=p/q$ as before. We will also denote by $\sopt$ some (unknown) least expanding $k$-subset, and let $\topt=N(\sopt)$ be its neighbor set (the set that $\sopt$ expands to). Note that the optimum expansion in such an instance is $|\topt|/k=r/d=r/k^{\alpha+\eps}$, and so to get an $O(k^{\alpha+(1-c)\eps})$-approximation, we need to find a set with expansion at most $O(r/k^{c\eps})$.

As we noted, the optimum neighbor set $\topt$ has average back-degree $d$ into $\sopt$. However, this might not be the case for all vertices in $\topt$. A common approach to avoiding non-uniformity for D$k$S and other problems is to prune small degree vertices. For example, a common pruning argument shows that a large fraction of edges is retained even after deleting all vertices with at most $1/4$ the average degree on each side (solely as a thought experiment, since we do not know $\sopt$ and $\topt$). However, the fundamental challenge in SSBVE is that we cannot delete small degree vertices in $V$, since this can severely skew the expansion of any set. A different and somewhat more subtle pruning argument, which explicitly bounds the ``penalty" we pay for small degree vertices in $\topt$, gives the following. 

\begin{claim}\label{clm:bad-neighbors} For every $\delta>0$, there exists a set $S'\subseteq\sopt$ of size at least $k/2$ and a set $T'\subseteq\topt$ with the following properties:
\begin{itemize}[itemsep=-2pt, topsep=4pt, partopsep=0pt]
  \item Every vertex $v\in T'$ has at least $\delta d/2$ neighbors in $S'$.
  \item Every vertex $u\in S'$ has at most $\delta r$ neighbors in $\topt\setminus T'$.
\end{itemize}
\end{claim}

\begin{proof}
  Consider the following procedure:
  \begin{itemize}
     \item Start with $S'=\sopt$ and $T'=\topt$.
     \item As long as $S'\neq\emptyset$ and the conditions are not met:
     \begin{itemize}
       \item Remove every vertex $v\in T'$ with at most $\delta d/2$ neighbors in $S'$ from $T'$.
       \item Remove every vertex $u\in S'$ with at least $\delta r$ neighbors in $\topt\setminus T'$ from $S'$.
     \end{itemize}
   \end{itemize}
  Call the vertices removed at iteration $t$ of the loop \emph{time $t$ vertices}. Note that time $t$ vertices in $\topt$ have at most $\delta d$ neighbors in $\sopt$ which are removed at some time $\geq t$, and time $t$ vertices in $\sopt$ have at least $\delta r$ neighbors in $\topt$ which were removed at time $\leq t$. Thus, if we look at the set of edges $E'=\{(u,v)\in E|u\in\sopt\setminus S',v\in\topt\setminus T',\mathrm{time}(u)\geq\mathrm{time}(v)\}$, we have \begin{align*}\delta r\cdot |\sopt\setminus S'|\leq|E'|&\leq (\delta d/2)|\topt\setminus T'|\\&\leq (\delta d/2)|\topt|=\delta kr/2.\end{align*}
Thus $|\sopt\setminus S'|\leq k/2$, and $|S'|\geq k/2$.
\end{proof}

For some constant $c>0$ to be determined later, let $S'$ and $T'$ be the sets derived from the above claim for $\delta=1/(2k^{c\eps})$. From now on, call a vertex $v\in\topt$ ``good" if $v\in T'$. Otherwise, call it ``bad".  Thus, we can restrict our attention to $S'$, allowing us to make the following simplifying assumption, and lose at most an additional constant factor in the approximation:
\begin{assumption}\label{asm:bad-neighbors} All vertices in $\sopt$ have at most $r/(2k^{c\eps})$ bad neighbors. 
\end{assumption}

Note that all the good vertices have at least $d/(4k^{c\eps})$ neighbors in $\sopt$. In addition, our assumption gives the following useful corollary:
\begin{corollary}\label{bad-expansion} Every set $S\subseteq\sopt$ with expansion at least $r/k^{c\eps}$ into some set $T$ has good expansion at least $r/(2k^{c\eps})$ into $T$.
\end{corollary}
\begin{proof} Every vertex in $S$ has at most $r/(2k^{c\eps})$ bad neighbors overall (and in particular in $T$), and so $S$ has at most $|S|r/(2k^{c\eps})$ bad neighbors in $T$. The rest must be good. \end{proof}

Finally, define $D:=n/k^{1-c\eps}$, and let $V_D=\{v\in V\mid \deg(v)\geq D\}$. Note that $nr=|E|\geq|V_D|D$, and so $$|V_D|\leq nr/D=rk^{1-c\eps},$$ and so every $k$-set in $U$ has expansion at most $r/k^{c\eps}$ into $V_D$. Thus,
while by Lemma~\ref{lem:at-most-exact} it suffices to find any set of size at most $k$ with expansion $O(k^{c\eps})$, it turns out that a weaker goal suffices:

\begin{claim}\label{clm:good-expansion} To find a $k$-subset of $U$ with expansion $\tilde O(r/k^{c\eps})$ (into $V$), it suffices to have an algorithm which returns a set of size at most $k$ with expansion $O(r/k^{c\eps})$ into $V\setminus V_D$.
\end{claim}
\begin{proof}
Let us examine the reduction which allows us to find a $k'$-subset of $U$ with small expansion, where $k'\leq k$, rather than choosing $k$ vertices in one shot (proving Lemma~\ref{lem:at-most-exact}): we repeatedly find such a set, and remove its vertices from $U$ until we have removed $k$ vertices. Note that the definition of $V_D$ may change as vertices are removed: at each such iteration, the degrees in $V$ may decrease, the value of $k$ decreases, and therefore the value of $D$ increases. However, these changes can only cause vertices to be removed from $V_D$, not added. Thus, every vertex which is in $V_D$ at some iteration was also in $V_D$ at the start of the algorithm.

If at each iteration we find a small set with small expansion into $V\setminus V_D$, then by the argument in the first writeup, this is sufficient to bound the total number of neighbors in $V\setminus V_D$  (of our $k$-subset of $U$) at the end of the algorithm, while losing only an additional $O(\log n)$ factor. Now, while the expansion into $V_D$ may have been very bad at any given iteration, the \emph{total} number of neighbors accrued throughout all iterations in $V_D$ (as defined at the start) is still at most $|V_D|\leq rk^{1-c\eps}$. Since $|\topt|\leq kr/d=k^{1-\alpha-\eps}r$, this gives us a $k^{\alpha+\eps-c\eps}$-approximation, as we wanted.
\end{proof}
Thus, we may shift our focus to finding small subsets with good expansion into $V\setminus V_D$ even if their expansion into $V_D$ is huge. 

\subsection{The algorithm}

Before we describe the algorithm, let us once again state that thanks to Claim~\ref{clm:good-expansion} our goal is simply to find a set $S\subseteq U$ of size at most $k$ with expansion at most $r/k^{c\eps}$ into $V\setminus V_D$.

Our algorithm will proceed according the same caterpillar construction described for the random planted setting, though unlike the random setting, each step will require a technically complex algorithm to ensure that either we maintain the same set size bounds as one would expect in a random instance, or we can abort the process and find a small set with small expansion into $V\setminus V_D$.

\subsubsection{First step}

Consider the following algorithm:
\begin{itemize}[itemsep=-2pt, topsep=4pt, partopsep=0pt]
  \item Let $U_D=\{u\in U : |N(u)\setminus V_D|\leq r/(2k^{c\eps})\}$. If $|U_D|=\Omega(k)$, return an arbitrary subset of $U_D$ of size $\min\{|U_D|,k\}$.
  \item Otherwise, guess a vertex $v\in V\setminus V_D$, and proceed to the next step with set $S=N(v)$.
\end{itemize}

\begin{lemma}\label{lem:first-step} The above algorithm either returns a set $S$ with the required expansion, or for at least one guess, returns a set $S\subseteq U$ such that $|S|\leq D$ and $|S\cap\sopt|\geq d/(4k^{c\eps})$.
\end{lemma}
\begin{proof}
  If all good vertices (in $\topt$) belong to $V_D$, then all vertices in $\topt\setminus V_D$ are bad. Since by Assumption~\ref{asm:bad-neighbors} every vertex in $\sopt$ has at most $r/(2k^{c\eps})$ bad neighbors, then for every $u\in\sopt$ we have $|N(u)\setminus V_D|\leq r/(2k^{c\eps})$. That is, $\sopt\subseteq U_D$. Thus, the first step returns a set of size $k$ with at most $k\cdot r/(2k^{c\eps})$ neighbors in $V\setminus V_D$, as required.

Otherwise, there exists a good vertex in $\topt\setminus V_D$. Thus, guessing such a vertex $v$ ensures that $S=N(v)$ has the desired properties by definition.
\end{proof}

\subsubsection{Hair step}

In the random planted setting, a hair step involves guessing a vertex $v\in\topt$, and replacing our current set $S$ with $S\cap N(v)$. In that setting, the change in the cardinality of $S$ and $S\cap\sopt$ is concentrated around the expectation by Chernoff bounds. However, in a worst-case setting, we need to ensure that the updated sets have the same cardinality as in the random planted setting in order to proceed to the next step. We do this using degree classification as in the first step, and the least-expanding-set algorithm. If either of these approaches gives a small set with small expansion, we are done. Otherwise, we show that we have the required cardinality bounds. Specifically, the algorithm for this step is as follows:

\begin{itemize}[itemsep=-0pt, topsep=4pt, partopsep=0pt]
  \item Given: A set $\hat U\subseteq U$, where $\hat{n}=|\hat U|$ and $\hat k=|\hat U\cap \sopt|\geq k^{1-\alpha}$.
  \item Let $\hat D=\hat n/k^{1-c\eps}$, and $\hat V_{\hat D}=\{v\in V : |N(v)\cap\hat U|\geq\hat D\}$, and $\hat U_{\hat D}=\{u\in \hat U : |N(u)\setminus \hat V_{\hat D}|\leq r/k^{c\eps}\}$.
  \item If $|\hat U_{\hat D}|\geq k$, return an arbitrary $k$-subset $S_{\hat D}\subseteq\hat U_{\hat D}$.
  \item Otherwise, run the least-expanding-set algorithm on the subgraph induced on $(\hat U_{\hat D}, V)$. If the resulting set has sufficiently small expansion, return this set.
  \item Otherwise, guess a vertex $v\in V\setminus \hat V_{\hat D}$, and proceed to the next step with set $U_v=N(v)\cap\hat U$.
\end{itemize}

The outcome of this algorithm and its analysis are somewhat similar to those of the first step, and captured by the following lemma.

\begin{lemma}\label{lem:hair-step} If $\hat k\geq k^{1-\alpha}$, then the above algorithm either returns a set $S$ with the required expansion, or for at least one guess, returns a set $U_v\subseteq \hat U$ such that $|U_v|\leq \hat n\cdot k^{c\eps}/k$ and $|U_v\cap\sopt|\geq \hat k\cdot k^{(1-c)\eps}/(2k^{1-\alpha})$.
\end{lemma}
\begin{proof}
First, note that $\hat nr=|E(\hat U,N(\hat U))|\geq |E(\hat U,\hat V_{\hat D})|\geq|\hat V_{\hat D}|\hat D$, and so $$|\hat V_{\hat D}|\leq \hat n r/\hat D=rk^{1-c\eps}.$$
Thus, if $|\hat U_{\hat D}|\geq k$ then every vertex in $S_{\hat D}$ has at most $r/k^{c\eps}$ neighbors in $V\setminus \hat V_{\hat D}$, and all together we get $rk^{1-c\eps}$ neighbors in this set. On the other hand, we get at most $|\hat V_{\hat D}|\leq rk^{1-c\eps}$ neighbors in $\hat V_{\hat D}$, so we get a total of $2rk^{1-c\eps}$ neighbors for $k$ vertices, and we are done.

Suppose, on the other hand, that $|\hat U_{\hat D}|\leq k$. Let $\hat S=\hat U\cap\sopt$ and $\hat T=N(\hat S)$. Note that since the $\hat k$ vertices in $|\hat S|$ have $\hat k r$ edges going into at most $kr/d$ vertices (in $\hat T\subseteq\topt$), the vertices in $\hat T$ have average back-degree at least $\hat k r/(kr/d)=\hat k d/k$ into $\hat S$. In this context, call the vertices in $\hat T$ with at least $\hat k d/2k^{1+c\eps}$ neighbors in $\hat S$ ``$\hat S$-good", and the rest ``$\hat S$-bad".

Similarly to Claim~\ref{clm:bad-neighbors}, it is easy to check that at least $\hat k/2$ vertices in $\hat S$ must have at most $r/k^{c\eps}$ $\hat S$-bad neighbors. Call this set $\hat S_{\mathrm{good}}$.

If all $\hat S$-good vertices (in $\hat T$) belong to $\hat V_{\hat D}$, then all vertices in $\hat T\setminus \hat V_{\hat D}$ are bad. Thus, by the above, for every vertex $u\in\hat S_{\mathrm{good}}$ we have $|N(u)\setminus \hat V_{\hat D}|\leq r/k^{c\eps}$. That is, $S_{\mathrm{good}}\subseteq \hat U_{\hat D}$. Note that $S_{\mathrm{good}}$ has at least $\hat k/2$ vertices and at most $|N(S_{\mathrm{good}})|\leq |\topt|=kr/d$ neighbors, and so it has expansion at most $2kr/(\hat kd)\leq 2k^{\alpha}r/d=2r/k^{\eps}<r/k^{c\eps}$ (where the first inequality follows since $\hat k\geq k^{1-\alpha}$). Since $|\hat U_{\hat D}|\leq k$, the least-expanding-set algorithm will return a set of size at most $k$ with at most the same expansion.

Otherwise, there exists an $\hat S$-good vertex in $\hat T_{\hat D}\setminus \hat V_{\hat D}$. Thus, guessing such a vertex $v$ ensures that $U_v=N(v)\cap\hat U$ has the desired properties by definition: $v$ is $\hat S$-good, so $$|N(v)\cap\hat U\cap\sopt|=|N(v)\cap\hat S|\geq \frac{\hat kd}{2k^{1+c\eps}}= \frac{\hat k\cdot k^{(1-c)\eps}}{2k^{1-\alpha}},$$ and $v\not\in \hat V_{\hat D}$, so $|N(v)\cap \hat U|\leq\hat D=\hat n\cdot k^{c\eps}/k$.
\end{proof}

\subsubsection{Backbone step}

In the random planted setting, a backbone step involves replacing our current set $S\subseteq U$ with $N(N(S))$. As in the hair step, the change in the cardinality of $S$ and $S\cap\sopt$ is concentrated around the expectation in the random planted setting, while in a worst-case setting, we need to ensure that the updated sets have the same cardinality as in the random planted setting in order to proceed to the next step. This is also done with the least-expanding-set algorithm between $S$ and $N(S)$. If this procedure gives a small set with small expansion, we are done. Otherwise, by binning vertices in $N(N(S))$ by degree, we produce $\log r$ sets, at least one of which we show will have the required cardinality bounds. The algorithm for this step is as follows:

\begin{itemize}[itemsep=-0.5pt, topsep=4pt, partopsep=0pt]
  \item Given: A set $\hat U\subseteq U$, where $\hat n=|\hat U|\leq k$ and $\hat k=|\hat U\cap\sopt|$.
  \item Let $\hat V=N(\hat U)\setminus V_D$.
  \item Run the least-expanding-set algorithm on the subgraph induced on $(\hat U,\hat V)$. If the resulting set has sufficiently small expansion, return this set.
  \item Otherwise, guess some $i\in[\lceil\log r\rceil]$, let $r_i=r/2^{i-1}$, and let $U_i=\{u\in N(\hat V)\mid |N(u)\cap\hat V|\in [r_i/2,r_i]\}$. Subsample this set by retaining every vertex independently with probability $r'/r$, and let $U'_i$ be the resulting set. Proceed to the next step with $U'_i$.
\end{itemize}

The guarantee of the backbone step algorithm is as follows.

\begin{lemma}\label{lem:backbone-step} If $\hat n\leq k$, then the above algorithm either returns a set $S$ with the required expansion, or for at least one guess, returns a set $U'_i\subseteq \hat U$ such that w.h.p.\ $|U'_i|\leq 2\hat n\cdot n^{\alpha}k^{c\eps}$ and $|U'_i\cap\sopt|\geq \hat k\cdot k^{\alpha+(1-2c)\eps}/(8\log r)$.
\end{lemma}
\begin{proof} Let $\hat S=\hat U\cap\sopt$ and let $\hat T=N(\hat S)\setminus V_D$.

First, suppose $|\hat T|\leq r\hat k/k^{c\eps}$. In this case, $\hat S$ has expansion at most $r/k^{c\eps}$ into $V\setminus V_D$, and the least-expanding-set algorithm will return a set of at most $|\hat U|\leq k$ vertices with at most the above expansion into $V\setminus V_D$, as required.

Otherwise, $|\hat T|\geq \hat rk/k^{c\eps}$, and $\hat S$ has expansion at least $r/k^{c\eps}$ into $V\setminus V_D$, and so by Corollary~\ref{bad-expansion}, $\hat S$ has \emph{good} expansion at least $rk/(2k^{c\eps})$ into $V\setminus V_D$. That is, it has at least $r\hat k/(2k^{c\eps})$ good neighbors in $\hat V\cap \topt$. Let us call this set of good neighbors $\hat T_{\mathrm{good}}$.

Consider the $\log r$ sets $\{U_i\cap\sopt\mid i\in[\lceil\log r\rceil]\}$. We know that at least one of them must cover at least a $1/\log r$-fraction of the edges between $\sopt\cap N(\hat V)$ and $\hat T_{\mathrm{good}}$. Choose~$i$ such that this is the case. Then since all the vertices in $T_{\mathrm{good}}$ are good, we know that
$$|E(\sopt\cap N(\hat V),\hat T_{\mathrm{good}})|\geq \frac{d}{4k^{c\eps}}\cdot|\hat T_{\mathrm{good}}|\geq \frac{d}{4k^{c\eps}}\cdot \frac{r\hat k}{2k^{c\eps}}.$$
On the other hand, we know that every vertex in $U_i\cap\sopt$ contributes at most $r'$ edges to this set, so we also have $|E(\sopt\cap N(\hat V),\hat T_{\mathrm{good}})|\leq|U_i\cap\sopt|r'$. Putting these together, we get the bound
$$|U_i\cap\sopt|\geq \frac{r}{r'}\cdot \frac{d\hat k}{8k^{2c\eps}}=\frac{r}{r'}\cdot \hat k\cdot \frac{k^{\alpha+(1-2c)\eps}}{8}.$$

Now consider $U_i$ itself. Since vertices in $U$ have degree $r$, we know that $|\hat V|\leq|N(\hat U)|\leq \hat nr$. Furthermore, since $\hat V\subseteq V\setminus V_D$, we know that every vertex in $\hat V$ has degree at most $D$, and so we have $$|E(U_i,\hat V)|\leq|E(N(\hat V),\hat V)|\leq |\hat V|D\leq \hat nrD.$$
On the other hand, every vertex in $U_i$ contributes at least $r'/2$ edges to this set, and so $|E(U_i,\hat V)|\geq |U_i|r'/2$. Putting these together, we get
$$|U_i|\leq \frac{2r}{r'}\cdot \hat nD=\frac{2r}{r'}\cdot \frac{\hat nnk^{c\eps}}{k}=\frac{2r}{r'}\cdot \hat nn^{\alpha}k^{c\eps}.$$

The required bounds on $|U'_i|$ and $|U'_i\cap\sopt|$ now follow from Chernoff bounds.
\end{proof}

\subsection{Putting everything together: the final step}

Before examining the final step, let us consider the effect of the first $q-1$ steps, assuming none of them stopped and returned a small set with small expansion, and assuming all the guesses were good (giving the guarantees in the various lemmas). Let $U_i$ be the set passed on from step $i$ to step $i+1$, and let $\hat n_i=|U_i|$ and $\hat k=|U_i\cap\sopt|$. Then to summarize the effect of the various steps, the first step gives
$$\hat n_1\leq n^{\alpha}k^{c\eps} \quad\text{and}\quad \hat k_1= \Omega(k^{\alpha+(1-c)\eps}),$$
a hair step $i$ gives
$$\hat n_i\leq \hat n_{i-1}\cdot\frac{k^{c\eps}}{n^{1-\alpha}}\quad\text{and}\quad\hat k_i= \hat k_{i-1}\cdot \Omega\left(\frac{k^{(1-c)\eps}}{k^{1-\alpha}}\right),$$
and a backbone step $i$ gives
$$\hat n_i= \hat n_{i-1}\cdot O\left(n^{\alpha}k^{c\eps}\right)\quad\text{and}\quad\hat k_i= \hat k_{i-1}\cdot \tilde\Omega\left(k^{\alpha+(1-2c)\eps}\right).$$

By induction, we can see that 
 after step $t$ we have
$$
\hat n_t=O(n^{\frac{tp}q-\lfloor\frac{tp}q\rfloor}\cdot k^{tc\eps})\quad\text{and}\quad
\hat k_t=\tilde\Omega(k^{\frac{tp}q-\lfloor\frac{tp}q\rfloor}\cdot k^{(t-(t+\lfloor\frac{tp}q\rfloor)c)\eps}).
$$

In particular, choosing $c>0$ such that $c<\min\{1/(q^2\eps),1/2\}$, this ensures the correctness of the assumptions $\hat n\leq k=n^{1-p/q}$ for backbone steps and $\hat k\geq k^{1-\alpha}=k^{1-p/q}$ for hair steps. When $t=q-1$, we get
\begin{align*}&\hat n_{q-1}=O(n^{1-p/q}\cdot k^{(q-1)c\eps})=O(k^{1+(q-1)c\eps})\quad\text{and}\quad\\&\hat k_{q-1}=\tilde\Omega(k^{1-p/q}\cdot k^{((q-1)-(p+q-2)c)\eps}).\end{align*}

Given $U_{q-1}$ with the above cardinality bounds, the final step is to simply run the least-expanding-set algorithm:
\begin{lemma} If $U_{q-1}$ has the above cardinality bounds, then running the least-expanding-set algorithm on $U_{q-1}$ and removing vertices arbitrarily to reduce the cardinality of the resulting set to $k$ gives us a subset of $U$ of size at most $k$ with expansion at most $$\tilde O\left(\frac{r}{k^{(q-(p+2q-3)c)\eps}}\right).$$
\end{lemma}
\begin{proof}
Note that $|U_{q-1}\cap\sopt|$ has expansion at most 
$$\frac{|\topt|}{|U_{q-1}\cap\sopt|}
=\tilde O\left(\frac{kr/d}{k^{1-\alpha}\cdot k^{((q-1)-(p+q-2)c)\eps}}\right)
=\tilde O\left(\frac{r}{k^{(q-(p+q-2)c)\eps}}\right).$$
Thus the least-expanding-set algorithm on $U_{q-1}$ will return a set with at most this expansion. However, this expansion will increase if we have more that $k$ vertices, and need to remove all but $k$ of them, without necessarily affecting the size of the neighborhood. By the above bound on $|U_{q-1}|=\hat n$, the maximum possible increase is by an $O(k^{(q-1)c\eps})$-factor, which would give us a $k$-subset of $U$ with at most the expansion guaranteed by the lemma. 
\end{proof}

Thus, to achieve the desired $r/k^{c\eps}$ upper bound on the expansion, it suffices to make sure that in addition to the above bounds on $c$, we also have $c<\frac{q}{p+2q-2}$.


\section{SDP and Sherali--Adam gaps}\label{sec:gaps}
In this section, we present a semidefinite programming (SDP) and Sherali--Adams (SA) integrality gap constructions for Small Set
Bipartite Vertex Expansion.
We prove that
\begin{itemize}
\item the SDP integrality gap is $\tilde O(\min(k,n/k))$ (where $\tilde O$ notation hides $\mathrm{polylog}(n)$ factors);\\ in particular, the gap is $\tilde O(\sqrt{n})$ for $k=\sqrt{n}$;
\item the SA integrality gap is $n^{1/4-O(\varepsilon)}$ after $r=\Omega(\varepsilon \log n/\log \log n)$ rounds \\(in this construction, $k = n^{1/2 - O(\varepsilon)}$).
\end{itemize}

We show both integrality gaps for the same family of instances --- random bipartite graphs $G = (U, V, E)$ with
$|U| = n$, $|V| = s \approx k$, in which every two vertices $u\in U$ and $v\in V$ are connected by an edge with probability $d_L/s$
(where $d_L = \Theta(\log n))$;
the expected degree of vertices in $U$ is $d_L$, and the expected degree of vertices in $V$ is $d_R = d_L n/s$.

\subsection{Integrality Gap Construction}
In this section, we describe the integrality gap instance, prove a lower bound on the cost of the optimal combinatorial solution, and
state some basic properties of the instance.

Given $n$, $k$, and $s$, we consider a random bipartite graph $G = (U, V, E)$ with $|U| = n$, $|V| = s$,
in which every two vertices $u\in U$ and $v\in V$ are connected with probability $d_L/s$, where $d_L \geq 20 \log_e n$.
\begin{lemma}\label{lem:basic-properties}
The following properties hold with probability at least $1 - O(1/n)$:
\begin{enumerate}
\item The cost of the optimal combinatorial solution is at least $\min(k,s)/2$.
\item Every vertex in $U$ has degree between $\frac{d_L}{2}$ and $\frac{3d_L}{2}$; every vertex in $U$ has degree
between $\frac{d_R}{2}$ and $\frac{3d_R}{2}$, where $d_R = n d_L /s$.
\item If $s \leq \sqrt{n}$, then every two vertices in $U$ are connected by a path of length 4 in $G$.
\end{enumerate}
\end{lemma}
\begin{proof}
\noindent 1. Note that we may assume that $k \leq s$: if $k > s$, we let $k' = s$ and prove that every subset $S'\subset U$
of size $k'$ has at least $\min(s,k')/2 = s/2$ neighbors in $V$; we get, $|N(S)| \geq |N(S')| \geq s/2$.

If the optimal solution has cost less then $k/2$, then there exist a subset $S\subset U$ of size $k$
and a subset $T\subset V$ of size $\lfloor k/2\rfloor$ such that $N(S)\subset T$ ($S$ is the optimal solution and $T$ is a superset of $N(S)$
of size $\lfloor k/2\rfloor$). Let us bound the probability that $N(S) \subset T$ for fixed sets $S$ and $T$ with $|S|=k$, $|T| = \lfloor k/2\rfloor$.

If $N(S) \subset T$, then there are no edges between $S$ and $V\setminus T$. There are at least
$ks/2$ pairs of vertices $(u,v)\in S\times (V\setminus T)$; the probability that there is no edge between any of them is
at most
$ (1 - d_L/s)^{ks/2} \leq e^{-d_L k/2}$.

There are at most $n^k$ ways to choose a subset $S$ of size $k$ in $U$ and
$s^{k/2} \leq n^{k/2}$ to choose a subset $T$ of size $\lfloor k/2\rfloor$ in $V$. By the union bound,
$\prob(N(S) \subset T \text{ for some } S\subset U, T\subset V, |S| = k, |T| = \lfloor k/2\rfloor)
\leq n^{\frac{3}{2}k} e^{- d_L k/2} = n^{\frac{3}{2}k - 10 k} = n^{-\Omega(k)}$.

\medskip
\noindent 2. Consider a vertex $u$. For every $v\in V$, let $\xi_v$ be the indicator variable for the event that
$(u,v) \in E$. Note that all random variables $\{\xi_v\}_{v\in V}$ are independent,
and $\Probab{\xi_v = 1} =d_L/s$.
By the Chernoff bound,
$$
\Probab{\deg u \notin \left[\frac{d_L}{2}, \frac{3d_L}{2}\right]} 
= \Probab{\Bigl|\sum_{v\in V} \xi_v - d_L\Bigr| > \frac{d_L}{2}} 
\leq 2 e^{-d_L/10}.
$$
By the union bound, 
$$
\prob\bigl(\deg u \notin [d_L/2, 3/2 d_L] 
\text{ for some }u\in U\bigr) 
\leq 2 e^{-d_L/10} n = O(1/n).
$$
Similarly, we show that
$$\Probab{\deg v \notin [d_R/2, 3/2 d_R] \text{ for some }v\in V} \leq  O(1/n).$$

\medskip

\noindent 3. Consider  $u_1,u_2\in U$. We assume that item 2 holds; in particular, $\deg u_1 \geq 2$ and
$\deg u_2 \geq 2$. Choose a random neighbor $v_1\in V$ of $u_1$, and a random neighbor $v_2\in V\setminus\{v_1\}$ of $u_2$.
We are going to show that $v_1$ and $v_2$ have a common neighbor $u'\in U\setminus \{u_1,u_2\}$ with high probability (given $v_1$ and $v_2$), and, therefore,
$u_1$ and $u_2$ are connected with a path $u_1\to v_1\to u'\to v_2 \to u_2$ of length 4 with high probability.
Note that events $\{(v_1, u') \in E\}$ for all $u'\in U\setminus \{u_1,u_2\}$,
events $\{(v_2, u') \in E\}$ for all $u'\in U\setminus \{u_1,u_2\}$ are independent.
Therefore, the probability that for a fixed $u'\in U$, $(v_1, u') \in E$ and $(v_2, u') \in E$ is $(d_L/s)^2 \geq  d_L^2/n$;
the probability that for some $u'\in U$, $(v_1, u') \in E$ and $(v_2, u') \in E$ is
$1 - (1 - d_L^2/n)^{n-2} \geq  1 - e^{-\Omega(d_L^2)} \geq 1 - O(1/n)$.
\end{proof}

\noindent In Section~\ref{sec:sdp}, we describe the standard SDP relaxation for SSBVE and prove that its value for the graph $G$ is $O(\max(d_L^2, d_Lks/n))$; therefore, for $k=s=n^{\delta}$ and $d_L = \Theta(\log n)$, the SDP integrality gap is
$\tilde O(\min(k,n/k))$ (see Theorem~\ref{thm:main-sdp}).
In Section~\ref{sec:SA}, we describe the SA relaxation and show that its gap is at least $n^{1/4 - O(\varepsilon)}$ after $\varepsilon \log n/\log \log n$ rounds (see Theorem~\ref{thm:LP-main}).

\subsection{SDP relaxation and solution}\label{sec:sdp}
In this section, we construct an SDP solution for a random bipartite graph $G$ described in the previous section.
It will be convenient for us to assume that the graph $G$ is biregular.
To make this assumption, we do the following.
First we choose $n$, $k$, $d_L$ and $d_R$ so that $d_L$ and $d_R$ are even integers.
By Lemma~\ref{lem:basic-properties}, $\deg u \leq d'_L = 3d_L/2$ and $\deg v \leq d'_R = 3d_R/2$ for every
$u\in U$ and $v\in V$.  We greedily add extra edges to the graph so that the obtained graph $G'$ is biregular with
vertex degrees $d'_L$ in $U$ and $d'_R$ in $V$. We construct an SDP solution for the graph $G'$;
this solution is also a feasible solution for the original graph $G$.
To simplify the notation, we denote the obtained graph by $G$ and the degrees of its vertices by $d_L$ and $d_R$.
Consider the following SDP relaxation for SSBVE.
\begin{align*}
\intertext{\textbf{SDP relaxation for SSBVE}}
\text{min }& \sum_{v \in V} \|\bar v\|^2\\
\text{s.t.}& \\
&\sum_{u\in U } \bar u = k \bar v_0\\
&\sum_{u\in U} \|\bar u\|^2 = k\\
&\langle \bar u, \bar v\rangle =\|\bar u\|^2 && \text{for all } (u,v) \in E, u\in U, v\in V\\
&\langle \bar w, \bar v_0 \rangle = \|\bar w\|^2 && \text{for all } w\in U \cup V\\
&\langle \bar w_1, \bar w_2 \rangle \geq 0 && \text{for all } w_1, w_2 \in U \cup V\\
&\|v_0\|^2 = 1
\end{align*}
Let us now state and prove the main result in this section, which implies Theorem~\ref{thm:SDP-main}.
\begin{theorem}\label{thm:main-sdp}
Assume that $d_L^2\ll k \ll n$ and $1\ll s$ (where $a \ll b$ denotes that $b/a \geq \log^2 n$).

The value of the SDP solution for $G$ is at most $4\max(d_L^2, d_Lks/n)$. In particular, when $k=s=n^{\delta}$ (with $\delta\in(0,1)$) and $d_L$ is $\mathrm{polylog}(n)$,
the SDP value is $\tilde O(\max(n^{2\delta-1},1))$ and the optimal combinatorial solution has value $\Omega(n^\delta)$. Thus, the gap is
$$\tilde \Omega(n^{\min(\delta,1-\delta)}) = \tilde \Omega(\min(k,n/k)),$$
\end{theorem}
\begin{proof}
Let $\nu_{u_1u_2}$ be the number of common neighbors of vertices $u_1,u_2\in V$. Let $\alpha = \frac{1}{2}\min(d_L n/(ks), 1)$.
Define matrices $A$, $B$, and $C$.
Let $A=(a_{u_1u_2})$ be a square $n\times n$ matrix with entries
$$
a_{u_1u_2} =
\begin{cases}
k/n, & \text{if } u_1 = u_2\\
\frac{\alpha k^2}{d_L (d_R-1)n} \nu_{u_1u_2} +
\frac{(1-\alpha)k^2-k}{n(n-1)}
, &\text{if } u_1\neq u_2
\end{cases}
$$
$B$ be an $s\times s$ matrix with diagonal entries equal to $\tau = 2d_L^2/(\alpha s) = 4\max(d_L^2/s, d_L k/n)$
and off-diagonal entries equal to $\tau/2$,
and  $C$ be an $n\times s$ matrix with
$$c_{uv} =
\begin{cases}
k/n,& \text{if } (u,v) \in E,\\
k(\tau - d_R/n)/(n - d_R), &\text{if } (u,v) \notin E.
\end{cases}
$$
Note that $k(\tau - d_R/n)/(n - d_R) > 0 $ since $\tau - d_R/n = 2d_L^2/(\alpha s) - nd_L/(ns) > 0$.

Finally, let $X = \begin{pmatrix} A & C\\C^T & B \end{pmatrix}$.
In Lemma~\ref{lem:psd}, we will show  that $X$ is positive semidefinite. Now we
describe how to construct an integrality gap assuming that $X$ is positive semidefinite.
Since $X$ is positive semidefinite, it is the Gram matrix of some vectors $\{\bar u, \bar v:u\in U, v\in V\}$; specifically,
\begin{align*}
\langle \bar u_1, \bar u_2\rangle &= a_{u_1u_2} && \text{for } u_1, u_2\in U\\
\langle \bar v_1, \bar v_2\rangle &= b_{v_1v_2} && \text{for } v_1, v_2 \in V\\
\langle \bar u, \bar v\rangle &= c_{uv} && \text{for } u\in U, v \in V.
\end{align*}
Define $v_0 = \frac{1}{k}\sum_{u\in U} \bar u$. We obtain an SDP solution. Let us verify that it satisfies all the SDP constraints.
\begin{enumerate}
\item Constraint $\sum_{u\in U } \bar u = k \bar v_0$ holds since $v_0 = \frac{1}{k}\sum_{u\in U} \bar u$.
\item We show that $\sum_{u\in U} \|\bar u\|^2 = k$. We have,
$$\sum_{u\in U} \|\bar u\|^2 = \sum_{u\in U} a_{uu} = \sum_{u\in U} \frac{k}{n} = k.$$
\item For every edge $(u,v)$, we have,
$$\langle \bar u, \bar v\rangle = c_{uv} = k/n = c_{uu} = \|\bar u\|^2,$$
as required.
\item For every $u\in U$, we have
\begin{align*}
\langle \bar u, \bar v_0 \rangle = \frac{1}{k} \sum_{u'\in U} \langle \bar u, \bar u'\rangle &= \frac{1}{k} \sum_{u'\in U} a_{uu'} \\
&{} =\frac{1}{k}\Bigl(\frac{k}{n} +
 \frac{(1-\alpha)k^2-k}{n(n-1)} (n-1) 
 +\frac{\alpha k^2}{d_L (d_R-1)n}\sum_{u'\in U\setminus\{u\}}\nu_{uu'}  \Bigr)\\
&{}=\frac{1}{k}\Bigl(\frac{k}{n} +\frac{(1-\alpha)k^2-k}{n}
+ \frac{\alpha k^2}{d_L (d_R-1)n} \cdot d_L(d_R-1)
\Bigr) \\&= \frac{k}{n}.
\end{align*}

For every $v\in V$, we have
\begin{align*}
\langle \bar v, \bar v_0 \rangle = \frac{1}{k} \sum_{u'\in U} \langle \bar v, \bar u'\rangle &= \frac{1}{k} \sum_{u'\in U} c_{vu'}\\
&=\frac{1}{k}\Bigl(\sum_{u'\in U:(u',v) \in E} \frac{k}{n}
+ \sum_{u'\in U:(u',v) \notin E} \frac{k(\tau - d_R/n)}{n - d_R}\Bigr)\\&
=\frac{1}{k}\left(d_R \times \frac{k}{n} + (n-d_R)\times  \frac{k(\tau - d_R/n)}{n - d_R}\right) \\
&= \tau = \|\bar v\|^2.
\end{align*}

\item Obviously, $\langle \bar w_1, \bar w_2 \rangle \geq 0$ since all entries of matrices $A$, $B$, and $C$ are non-negative.
\item Finally, we verify that $\|v_0\|^2 = 1$. We have,
$$ \|v_0\|^2 = \frac{1}{k} \sum_{u\in U} \langle \bar u, \bar v_0\rangle = \frac{1}{k} \sum_{u\in U} \|\bar u\|^2 = \frac{1}{k} \cdot n \cdot \frac{k}{n} = 1.$$
\end{enumerate}

The value of the SDP solution is
$$\mathsf{SDP} = \sum_{v\in V} \|\bar v\|^2 = s \tau = 4\max(d_L^2, d_L k s/n).$$
\end{proof}
\paragraph{Proof that $X$ is positive semidefinite.}
\begin{lemma}\label{lem:psd}
Matrix $X$ is positive semidefinite.
\end{lemma}
\begin{proof}
Now, we show that $X$ is positive semidefinite.
For every $v\in V$, define a symmetric $(n+s)\times (n+s)$ matrix $X^{(v)}$ with entries $X_{w_1w_2}^{(v)}$:
$$
x_{w_1w_2}^{(v)}=
\begin{cases}
\frac{\alpha k^2}{d_L(d_R-1)n}, &\text{if } w_1,w_2\in U \text{ and both } \\
&w_1 \text{ and } w_2 \text{ are adjacent to } v;\\ 
&\text{possibly, } w_1=w_2 \\
\frac{k}{n} - \frac{k(\tau-d_R/n)}{n-d_R}, &\text{if } w_1 \in U, w_2 =v \text{ and } \\
&w_1, w_2 \text{ are adjacent}\\
\frac{\tau}{2}, &\text{if } w_1 = w_2 = v\\
0, &\text{otherwise}
\end{cases}
$$
Note that all columns and rows for vertices $u\in V$ non-adjacent to $v$ consist of zeros; all
columns and rows for vertices $v'\in V\setminus\{v\}$ also consist of zeros. Ignoring these zero columns and rows $X^{(v)}$ equals
\newdimen\extrabelowhlineskip
\extrabelowhlineskip=\normalbaselineskip \advance\extrabelowhlineskip by 0.2em\relax
$$
\left(
\begin{array}{ccc|c}
\frac{\alpha k^2}{d_L(d_R-1)n} & \cdots & \frac{\alpha k^2}{d_L(d_R-1)n} & \frac{k}{n} - \frac{k(\tau-d_R/n)}{n-d_R}\\
\vdots& & \vdots &\vdots\\
\frac{\alpha k^2}{d_L(d_R-1)n} & \cdots & \frac{\alpha k^2}{d_L(d_R-1)n} & \frac{k}{n} - \frac{k(\tau-d_R/n)}{n-d_R}\\[0.4em]
\hline\rule{0pt}{\extrabelowhlineskip}
\frac{k}{n} - \frac{k(\tau-d_R/n)}{n-d_R} & \cdots  &
\frac{k}{n} - \frac{k(\tau-d_R/n)}{n-d_R} & \frac{\tau}{2}
\end{array}\right)
$$
Observe that $X^{(v)}$ is positive semidefinite if and only if the matrix
$$M_1 = \begin{pmatrix} \frac{\alpha k^2}{d_L(d_R-1)n} &  \frac{k}{n} - \frac{k(\tau-d_R/n)}{n-d_R}\\
 \frac{k}{n} - \frac{k(\tau-d_R/n)}{n-d_R} & \frac{\tau}{2}\end{pmatrix}$$
 is positive semidefinite. We verify that $M_1$ is positive semidefinite. First, the diagonal entries of $M_1$ are positive; second, its determinant is at least
 (recall that    $\tau = 2d_L^2/(\alpha s)$ and $d_R = nd_L/s$)
\begin{multline*}
\frac{\alpha k^2}{d_L(d_R-1)n}  \times \frac{\tau}{2} - \left(\frac{k}{n} - \frac{k(\tau-d_R/n)}{n-d_R}\right)^2\\
> \frac{\alpha k^2}{d_L(nd_L/s)n}\times\frac{2d_L^2}{\alpha s} - \frac{k^2}{n^2} = \frac{k^2}{n^2} - \frac{k^2}{n^2} = 0.
\end{multline*}
We conclude that $X^{(v)}$ is positive semidefinite.
Additionally, we define two matrices, a symmetric matrix $Y$ and diagonal matrix $Z$:
\begin{align*}
Y&=\left(
\begin{array}{ccc|ccc}
\frac{(1-\alpha) k^2-k}{n(n-1)} & \cdots & \frac{(1-\alpha) k^2-k}{n(n-1)} & \frac{k(\tau - d_R/n)}{n-d_R} & \cdots & \frac{k(\tau - d_R/n)}{n-d_R}\\
\vdots& &\vdots &\vdots&&\vdots\\
\frac{(1-\alpha) k^2-k}{n(n-1)} & \cdots & \frac{(1-\alpha) k^2-k}{n(n-1)} & \frac{k(\tau - d_R/n)}{n-d_R} & \cdots & \frac{k(\tau - d_R/n)}{n-d_R}\\[0.4em]
\hline\rule{0pt}{\extrabelowhlineskip}
\frac{k(\tau - d_R/n)}{n-d_R} & \cdots & \frac{k(\tau - d_R/n)}{n-d_R} & \tau/2&\cdots&\tau/2\\
\vdots& &\vdots &\vdots&&\vdots\\
\frac{k(\tau - d_R/n)}{n-d_R} & \cdots & \frac{k(\tau - d_R/n)}{n-d_R}& \tau/2&\cdots&\tau/2
\end{array}\right)
\\[0.5em]
Z&=\left(
\begin{array}{ccc|ccc}
\zeta &  &  & & & \\
& \ddots& &&&\\
&  & \zeta &  & \\[0.4em]
\hline\rule{0pt}{\extrabelowhlineskip}
&  & & 0&&\\
& & &&\ddots&\\
&  & & &&0
\end{array}\right)
\qquad \text{where } \zeta = \frac{k}{n}-\frac{\alpha k^2}{(d_R-1)n} - \frac{(1-\alpha)k^2-k}{n(n-1)}.
\end{align*}
We prove that matrices $Y$ and $Z$ are positive semidefinite. Observe that matrix $Y$ is positive semidefinite if and only if the following matrix $M_2$ is positive semidefinite
$$
M_2=\begin{pmatrix}
\frac{(1-\alpha)k^2-k}{n(n-1)} & \frac{k(\tau - d_R/n)}{n-d_R}\\[0.2em]
\frac{k(\tau - d_R/n)}{n-d_R} & \tau/2.
\end{pmatrix}.
$$
We verify that $M_2$ is positive semidefinite by computing its determinant (note that $\tau = 2d_L^2/(\alpha s) < 1/8$ and $\alpha \leq 1/2$):
\begin{align*}
&\frac{(1-\alpha)k^2-k}{n(n-1)} \times \frac{\tau}{2} - \left(\frac{k(\tau - d_R/n)}{n-d_R}\right)^2 \\
&\quad \geq
\frac{(1-\alpha)k^2-k}{n^2} \times \frac{\tau}{2} - \left(\frac{k\tau}{n}\right)^2\\
&\quad=\frac{k^2\tau}{2n^2}\left((1-\alpha) - \frac{1}{k}- 2\tau \right)\geq 0.
\end{align*}
 To verify, that $Z$ is positive semidefinite we check that $\zeta \geq 0$:
\begin{align*}
\zeta &= \frac{k}{n}-\frac{\alpha k^2}{(d_R-1)n} - \frac{(1-\alpha)k^2-k}{n(n-1)} \\
&\geq \frac{k}{n}\left(1-\frac{\alpha ks}{d_Ln} - \frac{(1-\alpha)k}{n-1}\right)\geq 0,
\end{align*}
 since $2\alpha ks \leq d_Ln$ and $k < (n-1)/2$.

  Finally, we prove that $X = Y+ Z + \sum_{v\in V} X^{(v)}$ and, therefore, $X$ is positive semidefinite.
We do that by showing that $x_{w'w''} = y_{w'w''}+ z_{w'w''} + \sum_{v\in V} x^{(v)_{w'w''}}$ for all $w'$ and $w''$.
Consider the possible cases:
$$
\setlength\arraycolsep{4pt}
\begin{array}{l||c|c|c||c}\hline\rule{0pt}{\extrabelowhlineskip}
\text{case} & \sum_{v\in V} x^{(v)}_{w'w''} & y_{w'w''} & z_{w'w''} &  x_{w'w''}\\[0.4em]
\hline\rule{0pt}{\extrabelowhlineskip}
w'=w''\in U & \frac{\alpha k^2}{d_L(d_R-1)n} d_L & \frac{(1-\alpha)k^2-k}{n(n-1)} & \frac{k}{n}-\frac{\alpha k^2}{(d_R-1)n} - \frac{(1-\alpha)k^2-k}{n(n-1)} & k/n\\[0.8em]
w'\neq w''\in U & \frac{\alpha k^2\cdot \nu_{w'w''}}{d_L(d_R-1)n}  & \frac{(1-\alpha)k^2-k}{n(n-1)} & 0 & \frac{\alpha k^2 \nu_{w'w''} }{d_L (d_R-1)n} +
\frac{(1-\alpha)k^2-k}{n(n-1)}\\[0.8em]
(w',w'')\in E& \frac{k}{n} - \frac{k(\tau-d_R/n)}{n-d_R} &\frac{k(\tau - d_R/n)}{n-d_R}&0&k/n\\[0.8em]
(w',w'')\in (U\times V) \setminus E& 0 &\frac{k(\tau - d_R/n)}{n-d_R}&0&\frac{k(\tau - d_R/n)}{n - d_R}\\[0.8em]
w'=w''\in V & \tau/2 & \tau/2 & 0 & \tau\\[0.8em]
w'\neq w''\in V & 0 & \tau/2 & 0 & \tau/2\\[0.4em]
\hline
\end{array}
$$
In all of these cases, we have $x_{w'w''} = y_{w'w''}+ z_{w'w''} + \sum_{v\in V} x^{(v)}_{w'w''}$.
\end{proof}

\subsection{Sherali--Adams gap}\label{sec:SA}

\begin{figure}[h]
\begin{align*}
\text{minimize } & \sum_{v\in V} x_v&&\\
\text{subject to:}&\\
&\sum_{u\in U} x_u \geq k&&\\
&x_v \geq x_u&& \text{for all }(u,v) \in E, u\in U, v\in V\\
&0 \leq x_w\leq 1&& \text{for all }w\in U \cup V
\\[1.2em]
\hline
\\[-0.2em]
\text{minimize } & \sum_{v\in V} x_{\set{v}}&&\\
\text{subject to:}&\\
&\sum_{u\in U} x_{S \cup \set{u},T} \geq k x_{S,T}&&\\
&x_{S \cup \set{v},T} \geq x_{S \cup \set{u},T}&& \text{for all }(u,v) \in E, u\in U, v\in V
\\
&0 \leq x_{S ,T}\leq 1&& \text{for all subsets }S, T \text{ s.t. }|S|+|T| \leq r+1\\
&x_{\varnothing} = 1.
\end{align*}
\caption{Basic and Sherali--Adams LP relaxations for SSBVE. In the SA relaxation,
the constraints hold for all subsets $S,T \subset U \cup V$ such that $|S|+|T|\leq r$, unless specified otherwise.}
\label{fig:SA-LP}
\end{figure}

In this section, we present an $O(n^{1/4-\varepsilon})$ gap for the Sherali--Adams relaxation for SSBVE
after $r=\Omega(\varepsilon \log n/\log \log n)$ rounds. Let us start with describing the basic linear programming (LP) and the lifted Sherali--Adams relaxations for the problem.

In the basic LP presented in Figure~\ref{fig:SA-LP}, we have a variable $x_w$ for every vertex $w\in U\cup V$.
In the Sherali-Adams LP, we have variables $x_S$ for all subsets $S\subset U \cup V$ of size at most $r+1$
and auxiliary variables $x_{S,T}$ that are linear combinations of variables $x_S$:
$$x_{S,T} = \sum_{J\subset T} (-1)^{|J|} x_{S\cup J}
$$
for subsets $S$ and $T$ such that $|S|+|T|\leq r+1$.
In intended integral solution,
$x_S = \prod_{w\in S} x_w$
and
$x_{S,T} = \prod_{w\in S} x_w\prod_{w\in T} (1-x_w)$.
Note that $x_{S,T} = 0$ if $S\cap T \neq 0$ and $x_{S,\varnothing} = x_S$.

As described above, let $G=(U, V, E)$ be a random bipartite graph with $|U|=n$, $|V| = s=\sqrt{n}$,
in which the expected degree of vertices in $U$ equals $d_L = \Omega(\log n)$
and the expected degree of vertices in $V$ equals $d_R = \sqrt{n} d_L$.
Let $k = n^{1/2 - O(\varepsilon)}$ (we will specify the exact of $k$ below).
We proved in Lemma~\ref{lem:basic-properties} that
 every two vertices in $U$ are connected with a path of length $4$ with high probability.
We will assume below that this statement holds.
For a set $S\subset U \cup V$, we denote $S_U = S\cap U$ and $S_V = (S \cap V) \setminus N(S_U)$.

\begin{definition}
Let us say that $({\cal T}, S')$ is a cover for a set $S\subset U \cup V$ if ${\cal T}$ is a tree in $G$ (possibly, ${\cal T}$ is empty), $S'\subset S_V$ (possibly, $S' = \varnothing$), and each vertex in $S_U \cup S_V$ lies in ${\cal T}$ or in $S'$;
we require that if ${\cal T}$ is not empty, it contains at least one vertex from $U$.
The cost of a cover $({\cal T}, S')$ is $|{\cal T}\cap U| +|S'| + 1$. A minimum cover of $S$ is a cover of minimum cost;
we denote the cost of a minimum cover by $\cost(S)$.
\end{definition}

Now we are ready to describe the Sherali--Adams solution.
We assume that $r \leq \varepsilon \log n/\log \log n$.
Let $\alpha = \frac{1}{2(r+1)}$, $\beta = \alpha^{r+1}  = n^{-O(\varepsilon)}$,
and $k = \beta \sqrt{n} /4$.

We define
$$x_S = \beta^{|S_U|} \alpha^{|S_V|} \frac{1}{n^{\cost(S)/4}}.$$
\begin{claim}\label{claim:monotone}
$\cost(S)$ is a non-decreasing function: $\cost(S) \leq cost(\tilde S)$ if $S\subset \tilde S$.
\end{claim}
\begin{proof}
Let $(\tilde{\cal T}, \tilde S')$ be a minimum cover for $\tilde S$. Then
$(\tilde{\cal T}, \tilde S'\cap S_V)$ is a cover for $S$ of cost at most $\cost(\tilde S)$.
\end{proof}
\begin{claim}\label{claim:exp-decay}
Consider a set $S$ of size at most $r$ and vertex $w\notin S$.
\begin{enumerate}
\item If $w\notin N(S_U)$ then $x_{S\cup\set{w}} \leq \alpha x_S$.
\item If $w\in N(S_U)$ then $x_{S\cup\set{w}} =  x_S$.
\end{enumerate}
(If $w\in S$, then trivially $x_{S\cup\set{w}} =  x_S$.)
\end{claim}
\begin{proof}
Denote $\tilde S = S \cup\set{w}$. We first prove that $x_{S\cup\set{w}} \leq \alpha x_S$
if $w\notin N(S_u)$.
Consider two cases.

\smallskip
\noindent\textbf{Case 1.} Assume that $w\in U$. Then $|\tilde S_U| = |S_U| +  1$.
We have,
\begin{multline*}
x_{\tilde S} = \beta^{|\tilde S_U|} \alpha^{|\tilde S_V|} n^{-\cost(\tilde S)/4} \leq
\beta^{|S_U|+1} n^{-\cost(S)/4} = \\ \beta^{|S_U|}\alpha^{r+1} n^{-\cost(S)/4}
\leq \alpha \beta^{|S_U|}\alpha^{|S_V|} n^{-\cost(S)/4} = \alpha x_S.
\end{multline*}
\noindent\textbf{Case 2.} Assume that $w\in V\setminus N(S_U)$. Then $w\in \tilde S_V$ and thus $|\tilde S_V| = |S_V| +  1$.
We also have $\tilde S_U = S_U$. Hence,
\begin{align*}
x_{\tilde S} &= \beta^{|\tilde S_U|} \alpha^{|\tilde S_V|} n^{-\cost(\tilde S)/4} \leq
\beta^{|S_U|} \alpha^{|S_V|+1} n^{-\cost(S)/4}  \\&= \alpha \beta^{|S_U|}\alpha^{|S_V|} n^{-\cost(S)/4}= \alpha x_S.
\end{align*}

Now we prove that $x_{S\cup\set{w}} =  x_S$ if $w\in N(S_U)$. Let $u$ be a neighbor of $w$ in $S_U$.
Consider a minimal cover $({\cal T}, S')$ for $S$. By the definition of a cover,
$u$ lies in $\cal T$. Define a tree $\tilde{\cal T}$ as follows:
if $w$ is not in $\cal T$, let $\tilde{\cal T} = {\cal T} + (w,u)$;
if $w$ is in $\cal T$, let $\tilde{\cal T} = \cal T$.
Then $(\tilde{\cal T}, S')$ is a cover for $\tilde S$.
It has the same cost as $({\cal T}, S')$ since $S_U \cap {\cal T} = \tilde S_U \cap \tilde{\cal T}$.
Therefore, $\cost(\tilde S) \leq \cost(S)$ and, by Claim~\ref{claim:monotone}, $\cost(S\cup \set{v}) = \cost(S)$.
Finally, observe that $S_U = \tilde S_U$ and $S_V = \tilde S_V$ (since $w\notin \tilde S_V$);
therefore, $x_S = x_{\tilde S}$.
\end{proof}
\begin{claim} \label{claim:SA-zero}
Consider sets $S$ and $T$, $|S| + |T| \leq r + 1$.
If $T\cap N(S) \neq \varnothing$ then $x_{S,T} = 0$.
\end{claim}
\begin{proof}
Assume that $T\cap N(S) \neq \varnothing$. Let $v$ be an arbitrary vertex in $T\cap N(S)$ and $u$ be its neighbor in $S$.
By Claim~\ref{claim:exp-decay}, item 2, $x_{S\cup J} = x_{S\cup J\cup \set{v}}$ for every subset $J \subset T\setminus\set{v}$.

We get that $x_{S,T}$ equals
 $$\sum_{J\subset T} (-1)^{|J|} x_{S\cup J}= \sum_{J\subset T\setminus\set{v}} (-1)^{|J|} \bigl(x_{S\cup J} - x_{S\cup J\cup \set{v}}\bigr) = 0.$$
\end{proof}
\begin{claim}\label{claim:SA-x-approx}
Consider disjoint sets $S$ and $T$ ($|S| + |T| \leq r+1$) such that $T\cap N(S) = \varnothing$.
We have,
$x_S/2 \leq x_{S,T} \leq x_S$.
\end{claim}
\begin{proof}
First, we show that $x_{S,T} \geq x_S/2$. We have,
\begin{align*}x_{S,T} &=\sum_{J\subset T} (-1)^{|J|} x_{S\cup J}\\
 &\geq \sum_{\substack{J\subset T\\|J| \text{ is even}}}
\Bigl(x_{S\cup J} - \sum_{w\in T\setminus J} x_{S\cup J \cup \set{w}}\Bigr)\\
&\hspace{-3mm}\stackrel{\text{\tiny Claim~\ref{claim:exp-decay}}}{\geq}
\sum_{\substack{J\subset T\\|J| \text{ is even}}}
\Bigl(x_{S\cup J} - \alpha (r+1) x_{S\cup J}\Bigr) \\ &\geq
\sum_{\substack{J\subset T\\|J| \text{ is even}}}
\frac{x_{S\cup J}}{2}\geq \frac{x_S}{2}.
\end{align*}
Second, we show that $x_{S,T} \leq x_S$. We have,
\begin{align*}x_{S,T} &=\sum_{J\subset T} (-1)^{|J|} x_{S\cup J} \\
&\leq x_S -
\sum_{\substack{J\subset T\\|J| \text{ is odd}}}
\Bigl(x_{S\cup J} - \sum_{w\in T\setminus J} x_{S\cup J \cup \set{w}}\Bigr)\\
&\hspace{-3mm}\stackrel{\text{\tiny Claim~\ref{claim:exp-decay}}}{\leq}
x_S - \sum_{\substack{J\subset T\\|J| \text{ is odd}}}
\Bigl(x_{S\cup J} - \alpha (r+1) x_{S\cup J}\Bigr) \leq
x_S
\end{align*}
\end{proof}
\begin{claim}
Let $S$ be a set of size at most $r$ and $u\in U$. Then
$$x_{S\cup\set{u}} \geq \frac{\beta}{\sqrt{n}} x_S.$$
\end{claim}
\begin{proof}
Let $\tilde S = S\cup\set{u}$. We have, $|\tilde S_U| \leq |S_U| + 1$ and $|\tilde S_V| \leq |S_V|$.

Now we upper bound $\cost(\tilde S)$. Consider a minimum cover $({\cal T}, S')$ for $S$. If $\cal T$ is empty,
let $\tilde{\cal T} = \{u\}$. Otherwise, let $u'\in U$ be a vertex from $\cal T$ (by the definition of a cover,
$\cal T$ contains vertices of $U$ if it is not empty). There is a path $P$ of length $4$ between $u$ and $u'$ in $G$:
$u\to v_1\to u'' \to v_2 \to u'$.  We connect $u$ to $\cal T$ with the path $P$ or its subpath (if one of the vertices,
$v_1$, $u''$, or $v_2$, is already in $\cal T$) and obtain a tree $\tilde{\cal T}$. Since we added at most two vertices
from $U$ to $\cal T$ (namely, vertices $u$ and $u''$), we have
$|\tilde{\cal T} \cap U| \leq |{\cal T} \cap U| + 2$. We get a cover $({\tilde{\cal T}}, S')$ for $\tilde S$ of cost at most
$\cost(S) + 2$.

Therefore,
\begin{align*}
x_{\tilde S} &= \beta^{|\tilde S_U|} \alpha^{|\tilde S_V|} \frac{1}{n^{\cost(\tilde S)/4}}\\
&\geq \beta^{|S_U|+1} \alpha^{|S_V|} \frac{1}{n^{(\cost(S)+2)/4}} = \frac{\beta}{\sqrt{n}}.
\end{align*}
\end{proof}
\begin{lemma}
The solution $x_{S,T}$, which  we presented, is a feasible solution for the Sherali--Adams relaxation.
Its value is at most $\alpha n^{1/4} < n^{1/4}$.
\end{lemma}
\begin{proof}
First, we verify that the $X_{S,T}$ is a feasible SDP solution. We check that $0\leq x_{S,T} \leq 1$. Clearly, $0\leq x_S \leq 1$ and $x_{\varnothing} = 1$.
By  Claims~\ref{claim:SA-zero} and~\ref{claim:SA-x-approx}, either $x_{S,T} = 0$ or
$0\leq x_S/2 \leq x_{S,T} \leq x_S \leq 1$.

Now consider an edge $(u,v)$ with $u\in U$ and $v\in V$.
By Claim~\ref{claim:exp-decay}, item 2,
$x_{S\cup J \cup\set{u,v}} = x_{S\cup J\cup \set{u}}$ for every subset $J\subset T$;
hence, $x_{S \cup\set{u,v}, T} = x_{S \cup\set{u}, T}$.
We have,
\begin{align*}
x_{S\cup\set{v},T} - x_{S\cup\set{u},T} &= x_{S\cup\set{v},T} - x_{S\cup\set{u,v},T}\\&
=x_{S\cup\set{v},T\cup \set{u}} \geq 0.\end{align*}
We show that $\sum_{u\in U} x_{S \cup \set{u},T} \geq k x_{S,T}$. We have,
\begin{align*}
\sum_{u\in U} x_{S \cup \set{u},T} & \geq \sum_{u\in U\setminus N(T)} x_{S \cup \set{u},T}\\
&\hspace{-3mm}\stackrel{\text{\tiny Claim~\ref{claim:SA-x-approx}}}{\geq} \sum_{u\in U\setminus N(T)} x_{S \cup \set{u}}/2\\
&\geq \frac{1}{2}\sum_{u\in U\setminus N(T)} \frac{\beta}{\sqrt{n}}  x_{S} \geq  \frac{\beta}{2\sqrt{n}}\times \frac{n}{2} \geq k.
\end{align*}

Finally, we note that $x_{\set{v}} = \alpha n^{-1/4}$ for every $v\in V$. Therefore,
the cost of the SDP solution is $\alpha n^{1/4} < n^{1/4}$.
\end{proof}

\noindent We now prove Theorem~\ref{thm:LP-main}.
\begin{proof}[Proof of Theorem~\ref{thm:LP-main}]
By Lemma~\ref{lem:basic-properties}, the value of the optimal combinatorial solution
is at least $\min(k,s)/2 = n^{1/2 - O(\varepsilon)}$; the value of the Sherali--Adams solution that we have constructed is
less than $n^{1/4}$. Therefore, the integrality gap is $n^{1/4 - O(\varepsilon)}$.
\end{proof}


\section{Small Set Vertex Expansion} \label{sec:SSVE}
In this section we prove Theorem~\ref{thm:SSVE-main}, giving a simple bicriteria approximation algorithm for the SSVE problem.  To do this, we will first state the following result from~\cite{LM14}, which we will use as a black box.

\begin{theorem}[\cite{LM14}]
There is an $O(\sqrt{\log n} \cdot \delta^{-1} \log \delta^{-1} \log \log \delta^{-1})$-bicreteria approximation algorithm for SSVE, where $\delta = k/n$.
\end{theorem}

Thus if $k \geq \sqrt{n}$, this algorithm already gives an $\tilde O(\sqrt{n})$-bicriteria approximation for SSVE.

Now suppose that this is not the case, i.e.~$k \leq \sqrt{n}$.  Recall the well-studied edge expansion version (Small Set Expansion).

\begin{definition}
In the \emph{Small Set Expansion} problem (SSE), we are given a graph $G = (V, E)$ and an integer $k \leq |V|$.  The goal is to find a set $S \subseteq V$ with $|S| \leq k$ which minimize $|E(S, \bar S)| / |S|$.
\end{definition}

SSE is conjectured to be hard to approximate within some constant (for appropriate $k$), but for our purposes we do not need such a strong approximation.  Either of the following upper bounds, due to Bansal et al.~\cite{sse} and R\"acke~\cite{Racke}, 
will suffice for us.

\begin{theorem}[\cite{sse}] \label{thm:sse}
There is an $O(\sqrt{\log n \log(n/k)}) \leq O(\log n)$-bicriteria approximation algorithm for SSE.
\end{theorem}

\begin{theorem}[\cite{Racke}] \label{thm:sse-racke}
There is an $O(\log n)$-approximation algorithm for SSE.
\end{theorem}

Let $S^*$ be the optimal solution to the SSVE instance, so it minimizes $|N(S^*) \setminus S^*| / |S^*|$ among sets of size at most $k$.  Let $L = N(S^*) \setminus S^*$.  Then the number of edges out of $S^*$ is $|E(S^*, \bar S^*)| \leq |L| |S^*|$.  Thus if we run the algorithm from Theorem~\ref{thm:sse}, we will get back a set $S$ with $|S| \leq k$ and $|E(S, \bar S)| / |S| \leq O(\log n) |L|$.  Since the edge expansion is an upper bound on the vertex expansion, this implies that
\begin{align*}
\frac{|N(S) \setminus S|}{|S|} &\leq \frac{|E(S, \bar S)|}{|S|} \leq O(\log n) |L| \\&\leq O(k \log n) \frac{|L|}{|S^*|} = O(k \log n)  \frac{|N(S^*) \setminus S^*|}{|S^*|}.
\end{align*}

Since $k \leq \sqrt{n}$, this implies that the set $S$ is a $\tilde O(\sqrt{n})$-approximation.  This completes the proof of Theorem~\ref{thm:SSVE-main}. 


\section{Discussion}\label{sec:discussion}
As noted earlier, our approximation guarantee for SSBVE is a function of $n=|U|$, which when transformed into M$k$U translates to an approximation guarantee which is a function of $m$, the number of hyperedges (sets).  In particular, we give an $m^{1/4 + \eps}$-approximation for arbitrarily small $\eps > 0$.  As we the goal is to minimize the number of elements (vertices) in the union, it may not seem natural to consider approximation guarantees which are a function of $m$. In particular, 
an arbitrary choice of $k$ sets clearly gives an $n$-approximation (OPT must be at least $1$ and the union of any $k$ sets is at most the universe size $n$).  So, for example, if $m \geq n^4$ then our algorithm does not beat this trivial algorithm. Unfortunately, as noted earlier, Conjecture~\ref{con:HDvR} rules out any polynomial time $O(n^{1-\eps})$ approximation for M$k$U for any constant $\eps > 0$, as we show in Appendix~\ref{sec:random-app}.

On a more positive note, consider the following problem, which we call \emph{SSVE-Union} (SSVE-U).  Given a graph $G = (V, E)$ and an integer $k$, the goal is to find the set $S \subseteq V$ with $|S| \leq k$ which minimizes $|N(S)| / |S|$.  This is clearly a special case of SSBVE, in which we have a bipartite graph with two copies of $V$ and the corresponding edges. Since both sides have the same cardinality, we get an $O(n^{1/4+\eps})$-approximation for this problem. 

Moreover, there is a simple reduction in the other direction: given an instance $(G(U,V,E),k)$ of SSBVE, construct a new graph in which all vertices in $V$ are attached to a sufficiently large clique. This shows that an $f(n)$-approximation for SSVE-U translates to an $f(n)$-approximation for SSBVE. In other words, the problems are computationally equivalent, and $n^{1/4+\eps}$ is always a non-trivial approximation guarantee for SSVE-U. 

An important remaining open problem concerns approximating M$k$U for hypergraphs with bounded edge size. Currently, the only tight approximations (according to Conjecture~\ref{con:DvR}) are for the graph case~\cite{BCCFV10, CDK12}, while even for 3-uniform hypergraphs, the current best known approximation~\cite{CDKKR16} (in terms of $n$, the number of vertices) is $\tilde O(n^{2/5})$, whereas the lower-bound matching Conjecture~\ref{con:HDvR} would be $n^{2-\sqrt{3}}\approx n^{0.2679}$. Plugging our main algorithm (for general hypergraphs) into the algorithm of~\cite{CDKKR16} for the 3-uniform case of M$k$U brings the approximation guarantee down to $O(n^{3/8+\eps})$, though this mild improvement is still far from the predicted lower bound, and it seems new techniques may be needed to achieve tight approximations for such special cases.

\bibliographystyle{plain}
\bibliography{refs}

\appendix


\section{Conditional Hardness Based on Random Models}\label{sec:random-app}

We examine approximation lower bounds arising from Conjecture~\ref{con:HDvR}. Let us start by seeing that the multiplicative gap between the M$k$U objective for the dense and random case tends to $m^{1/4}$ as the constant hyperedge size $r$ grows and $\eps$ decreases. Rather than computing the maximum possible gap for every $r$, let us consider a simple setting for which the conjecture predicts that we cannot distinguish: Let $\alpha=1-\eps$ and $\beta=1-2\eps$ (assume $\eps<1/r$). Since we are already using $k$ for the number of vertices in the ``dense" case, let our requirement for the number of edges in the subgraph be $\ell=k^{1+\beta}=\sqrt{m}$, giving us $k=m^{1/(4-4\eps)}$. Note that the edge probability in the random case is $n^{1+\alpha-r}=m/n^r=m^{1-r/(1+\alpha)}=m^{1-r/(2-\eps)}$. Thus, in the random case, for a $\hat k$-subgraph to have $\ell=m^{1/2}$ hyperedges, we need $$\hat k^rm^{1-r/(2-\eps)}=m^{1/2}\qquad\Rightarrow\qquad \hat k=m^{1/(2-\eps)-1/(2r)},$$ giving us a distinguishing gap of $$\hat k/k=m^{1/(2-\eps)-1/(4-4\eps)-1/(2r)}=m^{1/4-O(\eps^2)-O(1/r)}.$$

On the other hand, let us see that as a function of $n$, the conjecture gives us a distinguishing gap of $n^{1-\eps}$ for any $\eps>0$, ruling out the possibility of a good approximation as a function only of $n$. For sufficiently large $r$, let $\alpha=\sqrt{r}-1$, and $\beta=\sqrt{r}-1-\eps$. Also, let $k=n^{1/\sqrt{r}}$. 
Again, denote our requirement for the number of edges in the subgraph by $\ell=k^{1+\beta}=k^{\sqrt{r}-\eps}=n^{1-\eps/\sqrt{r}}$. Note that the edge probability in the random case is $n^{1+\alpha-r}=n^{\sqrt r-r}$. Thus, in the random case, for a $\hat k$-subgraph to have $\ell=n^{1-\eps/\sqrt{r}}$ hyperedges, we need $$\hat k^rn^{\sqrt r-r}=n^{1-\eps/\sqrt{r}}\qquad\Rightarrow\qquad \hat k=n^{1-1/\sqrt{r}+1/r-\eps/r^{3/2}},$$ giving us a distinguishing gap of $$\hat k/k=n^{1-2/\sqrt{r}+1/r-\eps/r^{3/2}}.$$


\section{Proof of Lemma~\ref{lem:preprocessing}}\label{sec:preprocessing-app}

Let us start with the $r$-uniformity claim. Note that the vertices of $U$ can be partitioned into $\log |V|$ buckets $B_1,\ldots,B_{\lceil\log |V|\rceil}$ such that the degree of every node in $B_i$ is in the range $[2^{i-1},2^i]$. Let $\sopt$ be an optimum $k$ subset of $U$. Then at least one bucket $B_i$ contains a $1/\log|V|$ fraction of the nodes in $\sopt$, and thus has expansion at most $$\frac{|N(\sopt)\cap B_i|}{|\sopt\cap B_i|}\leq\frac{|N(\sopt)|}{|\sopt\cap B_i|}\leq\log|V|\cdot\frac{|N(\sopt)|}{|\sopt|}.$$
Thus, since it suffices to approximate the least expanding set of cardinality \emph{at most} $k$, it suffices to focus on the subgraph induced on $(B_i,V)$ (up to a $\log V$ factor). However, this does not make the left nodes $r(=2^i)$-uniform. This can be easily fixed by adding $r$ nodes to $U$, and adding $2^{i}-\deg(u)$ edges from every $u\in B_i$ to this set. It is easy to see that this does not affect the expansion of any set in $B_i$ by more than a constant factor.

The assumption that the cardinality of $\topt=N(\sopt)$ is known is easily justified by trying all possible values $\{r,\ldots,|V|\}$ in our algorithm, and returning the least expanding set from among these $|V|-r+1$ iterations. As noted, the average ``back-degree" (the number of neighbors in $\sopt$) among vertices in $\topt$ is then known to be $d=kr/|\topt|$.

Let $\alpha=\log_n(n/k)$, so that $k=n^{1-\alpha}$. Our goal then is to obtain a $k^{\alpha+\eps}$-approximation (for arbitrarily small constant $\eps>0$). Again, note that the trivial algorithm of choosing an arbitrary $k$-subset of $U$ (which may have at most $kr$ neighbors) gives a $d$ approximation. Thus we may assume that $d\geq k^{\alpha+\eps}$.

In fact, in our analysis we will \emph{require} $d$ to be $k^{\alpha+\eps}$ for some sufficiently small $\eps$, so if $d$ is much larger, we need to somehow reduce it. This is accomplished via a subsampling idea which appeared in \cite{Bhaskara-thesis} (in the context of D$k$S): If $d=k^{\alpha+\eps+\beta}$ for some $\beta\in(0,1-\alpha-\eps]$, prune $U$ by retaining every $u\in U$ independently with some probability $n^{-\gamma}$, and let $U_\gamma$ be the remaining vertices in $U$. Note that $|U_\gamma|$ is tightly concentrated around $n_\gamma:=n^{1-\gamma}$, and similarly, $$k_\gamma:=|\sopt\cap U_\gamma|\approx kn^{-\gamma}=n^{1-\alpha-\gamma}\approx n_\gamma^{1-\frac{\alpha}{1-\gamma}}.$$
In particular, $\topt$ has average back-degree roughly $dn^{-\gamma}$ into $\sopt\cap U_\gamma$.
Furthermore, a Chernoff bound argument gives the following:
\begin{claim} With high probability, for any set $S\subseteq U_\gamma$ and $T=N(S)$, if $T$ has average back-degree $d'$ into $S$ for some sufficiently large $d'$ (at least $log(n)$), then there is a set $S'\subseteq U$ such that $|S'|\geq |S|n^{\gamma}$ and $N(S')=T$.
\end{claim}

Thus, if we can choose $\gamma$ such that the new optimal back-degree $dn^{-\gamma}$ is $k_\gamma^{\frac{\alpha}{1-\gamma}+\eps}$, 
 then a $k_\gamma^{\frac{\alpha}{1-\gamma}+\eps}/f$-approximation\footnote{This is consistent with our general plan. When $k_\gamma=n_\gamma^{1-c}$ for some constant $c$,  and the optimal back-degree is $k_\gamma^{1+\eps}$, we want something that is roughly a $k_\gamma^{c+\eps}$ approximation.} for some $f\geq\log n$ will give us a solution with back degree at least $f$, from which we can recover a solution in the original graph with back-degree \begin{align*}f\cdot n^{\gamma}=f\cdot d/k_\gamma^{\frac{\alpha}{1-\gamma}+\eps}&\approx f\cdot d/n^{\alpha-\frac{\alpha^2}{1-\gamma}+(1-\alpha-\gamma)\eps}\\&>f\cdot d/n^{\alpha(1-\alpha)+\eps}.\end{align*}
Thus, it remains to show that there is indeed some $\gamma$ such that $dn^{-\gamma}=k_\gamma^{\frac{\alpha}{1-\gamma}+\eps}$. Note that the left hand side is larger for $\gamma=0$. However, by increasing $\gamma$, equality must be achieved at some point, since raising $\gamma$ so that $n^{\gamma}=d$, we get $1$ on the left, and $(k/d)^{\frac{\alpha}{1-\gamma}+\eps}<1$ on the right.

Thus, we will assume without loss of generality that $d=k^{\alpha+\eps}$ for some arbitrarily small constant $\eps>0$, and our goal will be to find a $k^{\alpha+\eps}/f$-approximation for some $f\geq\log n$. Since we will indeed need $\eps>0$ to be a constant (see the next paragraph), we may as well take larger $f$ to improve our approximation guarantee. 
So our goal will be to obtain a $k^{\alpha+\eps-c\eps}$-approximation for some small constant $c>0$.

Finally, our algorithm will be tailored to parameters of the form $k=n^{1-\alpha}$ where $\alpha=p/q$ for some bounded (mutually prime) integers $q>p>0$. This can also be achieved as above by choosing $\gamma$ carefully (so that $k_\gamma=n_\gamma^{1-p/q}$ for some bounded integers $p,q$). In this case we will have $\eps=O(1/q)$.

\end{document}